\documentclass[sigplan, screen, authorversion]{acmart}
\usepackage{Style-Files/preamble}

\usepackage{amsmath}
\usepackage{xspace}
\usepackage{filecontents}
\usepackage{algorithm}
\usepackage{amsthm}
\usepackage[noend]{algpseudocode}
\usepackage{mathtools}
\usepackage{mdframed}
\usepackage{enumitem}
\usepackage{xcolor}
\usepackage{xurl}
\usepackage{mathrsfs}

\usepackage{watermark} 

\setcopyright{rightsretained}
\copyrightyear{2021}
\acmYear{2021}
\acmDOI{10.1145/3477132.3483552}

\acmConference[SOSP '21]{ACM SIGOPS 28th Symposium on Operating Systems Principles}{October 26--29, 2021}{Virtual Event, Germany}
\acmBooktitle{ACM SIGOPS 28th Symposium on Operating Systems Principles (SOSP '21), October 26--29, 2021, Virtual Event, Germany}
\acmPrice{15.00}
\acmISBN{978-1-4503-8709-5/21/10}

\begin{document}

\date{}


\title{\sys: Breaking up BFT with ACID (transactions)}
 

\author{\large {\rm Florian Suri-Payer}, {\rm Matthew Burke}, {\rm Zheng Wang}, {\rm Yunhao Zhang}, {\rm Lorenzo Alvisi}, {\rm Natacha Crooks$^\dagger$} \\
\vspace{2pt}\normalsize {\it Cornell University}, {\it $^\dagger$UC Berkeley}}

\renewcommand{\shortauthors}{Suri-Payer, et al.}




\begin{abstract}
\vspace{0.5em}

This paper presents~\sys{}, the first transactional,
leaderless Byzantine Fault Tolerant key-value store. 
\sys{} leverages ACID transactions to \textit{scalably} implement the abstraction of
a trusted shared log in the presence of Byzantine actors. Unlike
traditional BFT approaches, \sys{} executes non-conflicting operations
in parallel and commits transactions in a single round-trip during fault-free executions.
\sys{} improves throughput over traditional
BFT systems by four to five times, and is only four times slower
than TAPIR, a non-Byzantine replicated system.
\sys{}'s novel recovery mechanism further minimizes the impact of failures: with 30\% Byzantine clients, throughput drops by less than
25\% in the worst-case.

\end{abstract}

\maketitle
\pagestyle{empty}
\section{Introduction} 

This paper presents \sys{}{\footnote{\changebars{Despite (or because of) his deeply Fawlty character, Basil managed to rise first from paesant to Byzantine emperor (867-886) and then to hotel owner (1975-1979).}}{} a leaderless transactional key-value store that scales the abstraction of a Byzantine-fault tolerant shared log.

Byzantine fault-tolerance (BFT) systems enable safe online data sharing among mutually distrustful parties, as they guarantee correctness in the presence of malicious (Byzantine) actors. These platforms offer exciting opportunities for a variety of applications, including healthcare~\cite{HyperledgerHealth}, financial services~\cite{IBMSettlements, Diem, StateFarmSubrogation}, and supply chain management~\cite{IBMFoodSupply, DeloitteSupply}. 
\changebars{One could for instance design a fully decentralized payment infrastructure between a consortium of banks that omits the need for current centralized automated clearing houses \cite{Diem}. None of the participating banks may fully trust one another, yet they must be willing to coordinate and share resources to provide the joint service.}{Consider the iPhone supply chain: it spans three continents and hundreds of contractors \cite{AppleSup} who may not trust each other, yet must be willing to share product information.} BFT replicated state machines~\cite{castro1999practical,KotlaCACM,clement09aardvark, gueta2018sbft, yin2019hotstuff} and permissioned blockchains~\cite{Hyperledger,EthereumQuorum, buchman2016tendermint, al2017chainspace,kokoris2018omniledger,gilad2017algorand, baudet2019state} are at
the core of these new services: they ensure that mutually distrustful parties produce the same totally ordered log of operations.

The abstraction of a totally ordered log is appealingly simple. A {\em scalable} totally ordered log, however, is not only hard to implement (processing all requests sequentially can become a bottleneck), but also often unnecessary. Most distributed applications primarily consist of logically concurrent operations; supply chains for instance, despite
their name, are actually complex networks of independent transactions. 

Some BFT systems use sharding to try to tap into this parallelism. Transactions that access disjoint shards can execute concurrently, but operations within each shard are still totally ordered. Transactions involving multiple shards are instead executed by running cross-shard atomic commit protocols, which are layered above these totally ordered shards~\cite{kokoris2018omniledger, al2017chainspace, zamani2018rapidchain, padilha2013augustus, padilha2016callinicos}. 
The drawbacks of \changebars{systems that adopt this archtitecture}{these systems} are known: \one~they pay the performance penalty of  redundant coordination---both across shards (to commit distributed transactions) and among the replicas within each shard (to totally order in-shard operations) ~\cite{zhang2016operation, zhang2015tapir, mu2016consolidating}; \two~within each shard,  they give a leader replica undue control over the  total order ultimately agreed upon, raising fairness concerns~\cite{herlihy2016enhancing,
yin2019hotstuff, zhang2020byzantine}; \three~and often they restrict the expressiveness of the transactions they support~\cite{padilha2013augustus, padilha2016callinicos} by requiring that \changebars{their}{the} read and write set be known in advance.

In this paper, we advocate a more principled, performant, and expressive approach to supporting the abstraction of a totally ordered log at the core of all
permissioned blockchain systems. We make our own the lesson of distributed databases, which successfully leverage generic, interactive transactions to implement the abstraction of a sequential, general-purpose log. These systems specifically design highly concurrent protocols that are \textit{equivalent} to a serial schedule~\cite{bernstein1979fas,Papadimitriou1979serializability}. 
Byzantine data processing systems need be no different: rather than aiming to sequence all
operations, they should decouple the {\em abstraction} of a totally ordered sequence of transactions from its {\em implementation}. Thus, we flip the conventional approach: instead of building database-like transactions on top of a sharded, totally ordered BFT log, we directly build out this log
abstraction above a partially-ordered distributed database, where total order is demanded only for conflicting operations.

To this effect, we design \sys{}, a serializable BFT key-value store that implements the abstraction of a trusted shared log, whose novel design addresses each of the drawbacks of traditional BFT systems: \one~it borrows databases' ability to leverage concurrency control to support highly concurrent but serializable transactions, thereby adding parallelism to the log; \two~it sidesteps concerns about the fairness of leader-based systems by giving clients the responsibility of driving the execution of their own transactions; \three~it eliminates redundant coordination by integrating distributed commit with
replication~\cite{zhang2015tapir,mu2016consolidating}, so that, in the absence of faults and contention, transactions can return to clients in a single round trip; and ~\four it improves the programming API, offering support for general interactive transactions that do not require a-priori knowledge of reads and writes.

We lay the foundations for \sys{} by introducing two complementary notions of correctness. \textit{Byzantine isolation} focuses on safety: it 
ensures that correct clients observe a state of the database that could have been produced by correct clients alone. \textit{Byzantine independence} instead safeguards liveness: it limits the influence of Byzantine actors in determining whether a transaction  commits or aborts.
To help enforce these two notions, and disentangle correct clients from the maneuvering of Byzantine actors, 
\sys{}'s design follows the principle of \textit{independent operability}: it enforces safety and liveness through mechanisms that operate on a per-client and per-transaction basis. Thus, \sys{} avoids mechanisms that  enforce isolation through pessimistic locks (which would allow a Byzantine lock holder to prevent the progress of other transactions), adopting instead an optimistic approach to concurrency control.

Embracing optimism in a Byzantine setting comes with its own risks. Optimistic concurrency control \changebars{(OCC)}{}  protocols~\cite{kung1981occ,bernstein1983multiversion,reed1983atomic,xie2015callas,zhang2015tapir} are intrinsically vulnerable to aborting transactions if they interleave unfavorably during validation, and Byzantine faults can compound this vulnerability. Byzantine actors may, for instance, intentionally return stale data, or collude to sabotage the commit chances of correct clients' transactions. Consider
multiversioned timestamp ordering (MVTSO)~\cite{bernstein1983multiversion,reed1983atomic}, which allows writes to become visible to other operations before a transaction commits. While this choice  helps reduce abort rates for contended workloads, it can cause transactions to stall on uncommitted operations.

\sys{}'s ethos of independent operability is key to mitigating this issue. 
The system implements a variant of MVTSO that prevents Byzantine participants from unilaterally aborting correct clients' transactions and ncludes a novel \textit{fallback} mechanism that empowers clients to finish pending transactions issued by others, while preventing Byzantine actors from dictating their outcome.  Importantly, this fallback is a per-transaction recovery mechanism: thus, unlike traditional BFT view-changes, which completely suspend the normal processing of all operations, it can take place without blocking non-conflicting transactions.

Our results are promising: on TPC-C~\cite{tpcc}, Retwis~\cite{retwis} and Smallbank~\cite{difallah2013oltp}), \sys{}'s throughput is  
3.5-5 times higher than layering distributed commit over totally ordered shards running  BFT-SMaRt, a state-of-art PBFT implementation~\cite{bessani2014state} and HotStuff~\cite{yin2019hotstuff} (Facebook Diem's core consensus protocol \cite{baudet2019state}). 
BFT's cryptographic demands, however, still cause \sys to be 2-4 times slower than TAPIR, a recent non-Byzantine distributed database~\cite{zhang2015tapir}. 
In the presence of Byzantine clients, \sys{}'s performance degrades gracefully: with 30\% Byzantine clients, \changebars{the throughput of \sys{}'s correct clients}{\sys{}'s throughput} drops by less than 25\% in the worst-case. 
In summary, this paper makes the following three contributions: 
\begin{itemize}
\item It introduces the complementary correctness notions of Byzantine isolation and Byzantine independence.

\item It presents novel concurrency control, agreement, and fallback protocols that balance the desire for high-throughput in the common case with resilience to Byzantine attacks.  \item It describes \sys{}, a BFT database that guarantees Byz-serializability while preserving Byzantine independence. \sys {} supports interactive transactions, is leaderless, and achieves linear communication complexity.
\end{itemize}

\section{Model and Definitions}
\label{section:model}
We introduce the complementary and system-independent notions of Byzantine isolation and Byzantine independence, which, jointly formalize the degree to which a Byzantine actor can affect transaction progress and safety.
\subsection{System Model}
\label{s:Model}
\sys inherits the standard assumptions of prior BFT work. 
\changebars{A participant is considered correct if it adheres to the protocol specification, and faulty otherwise. Faulty clients and replicas may deviate arbitrarily from
their specification; a strong but static adversary can
coordinate their actions but cannot break standard cryptographic primitives. A shard contains a partition of the data in the system.}{} \changebars{We assume}{It assumes} that the number of faulty replicas
within a shard does not exceed a threshold \textit{f} and that an arbitrary number of clients may be faulty; \changebars{}{But for this bound, } we make no further assumption about the pattern of failures across shards.
We assume that applications authenticate clients and can subsequently audit their actions. 
\changebars{}{Faulty clients and replicas may deviate arbitrarily from
their specification; a strong but static adversary can
coordinate their actions but cannot break standard cryptographic primitives.}
Similar to other BFT systems~\cite{castro1999practical, fischer1985impossibility, KotlaCACM,
clement09aardvark, buchman2016tendermint}, \sys{} makes no synchrony assumption for safety 
but for liveness~\cite{Fischer85Impossibility} depends on partial synchrony~\cite{dwork1988consensus}.
\sys{} also inherits some of the limitations of prior BFT systems: it cannot prevent authenticated Byzantine clients, who otherwise follow the protocol, from overwriting correct
clients' data. It additionally assumes that, collectively,  Byzantine and correct clients have similar processing capabilities, and thus Byzantine clients cannot cause a denial of service attack by 
flooding the system.

\subsection{System Properties}
To express \sys's correctness guarantees, we introduce the notion of  {\em Byzantine isolation}. Database isolation (serializability, snapshot isolation, etc.) traditionally regulates the interaction between concurrently executing transactions;  Byzantine isolation ensures that, even though Byzantine actors may choose to violate ACID semantics, the state observed by correct clients will always be ACID compliant.

We start from the standard notions of transactions and histories introduced by Bernstein et al.~\cite{bernstein1987concurrency}. We summarize them
here and defer a more formal treatment to the Appendix~\ref{sec:proofs}. A transaction \textit{T}  contains a sequence of read and write operations terminating with a commit or an abort.  A history \textit{H} is a partial order of operations representing the interleaving of concurrently executing transactions, such that all conflicting operations are ordered with respect to one another. Additionally, let \textit{C}
be the set of all clients in the system; $\mathit{Crct} \subseteq C$ be the set
of all correct clients; and $Byz \subseteq C$ be the set of all
Byzantine clients. 
A projection $H|_{\mathscr{C}}$ is the subset of the partial order of operations in $H$ that were issued by the set of clients $\mathscr{C}$. We further adopt standard definitions of database isolation: a history
satisfies an isolation level I if the set of operation interleavings in H is allowed by I. Drawing from the notions of BFT linearizability~\cite{liskov2006tolerating} and view serializability~\cite{bernstein1987concurrency}, we then define the following properties:

\par\textbf{Legitimate History} History $H$ is \textit{legitimate} if it was generated by correct participants, {\em i.e.}, $H = H_\mathit{Crct}$.

\par\textbf{Correct-View Equivalent} History $H$ is \textit{correct-view}
equivalent to a history $H'$ if all operation results, commit decisions, and
final database values in $H|_{Crct}$ match those in $H'$.

\par \textbf{Byz-I} Given an isolation level $I$,
a history $H$ is \textit{Byz-I} if there exists a legitimate history $H'$ such that $H$ is correct-view equivalent to $H'$ and $H'$ satisfies $I$.

This definition is not \sys-specific, but captures what it means, for any  Byzantine-tolerant database,  to enforce the guarantees offered by a given isolation level $I$.  Informally, it requires that the states observed by correct clients be explainable by a history that satisfies I and involves only correct participants. It intentionally makes no assumptions on the states that Byzantine clients choose to observe.

\sys{} specifically guarantees Byz-serializability: correct clients will observe a sequence of states that is consistent with a sequential execution of concurrent transactions.  This is a strong safety guarantee, but it does not enforce application progress; a Byz-serializable system could still allow Byzantine actors to systematically abort all transactions. We thus define the notion of \textit{Byzantine independence}, a general system property that bounds the influence of Byzantine participants on the outcomes of correct clients' operations.

\par \textbf{Byzantine Independence} For every operation $o$ issued by a correct
client $c$,  no group of participants containing solely Byzantine actors can
unilaterally dictate the result of $o$.

In a context where clients issue transaction operations, Byzantine independence implies, for instance, that Byzantine actors cannot collude to single-handedly abort a correct client's transaction. This is a challenging property to enforce. It cannot be attained in a leader-based system: if the leader and a client are both Byzantine, they can collude to prevent a transaction from committing by strategically generating conflicting requests. In contrast, \sys{} can enforce  Byzantine independence
as long as Byzantine actors do not have full control of the network, a requirement that is in any case a precondition for any BFT protocol that relies on partial synchrony~\cite{fischer1985impossibility, miller2016honey}. We prove in \tr{Appendix \ref{sec:proofs}}{our supplemental material} that:

\par \textbf{Theorem} \ref{proof:ser}. {\em \sys{} maintains \textit{Byz-serializability}.}

\par \textbf{Theorem} \ref{proof:independence}. {\em \sys{} maintains Byzantine independence in the absence of network adversary.}

\sys{} is designed for settings where Byzantine attacks can
occur, but are infrequent, consistent with the prevalent assumption for permissioned blockchains today; namely,  that to maintain standing in a
permissioned system, clients are unlikely to engage in actively detectable
Byzantine behavior~\cite{haeberlen2010accountable} and, if they cannot break safety undetected, it is preferable for them to be live~\cite{malkhi2019flexible}.  We design \sys{} to be particularly efficient during gracious executions~\cite{clement09aardvark} ({\em i.e.}, synchronous and fault-free) while bounding overheads when misbehavior does occur. In particular, we design aggressive concurrency control mechanisms that maximize common case performance by optimistically exposing uncommitted operations, but ensure that these protocols preserve independent operability, so that \sys{} can guarantee
continued  progress under Byzantine attacks~\cite{clement09aardvark}. We confirm this experimentally in Section~\ref{section:eval}.

\section{System Overview}
\label{section:arch}

\begin{figure}[!th]

\includegraphics[width=0.48\textwidth]{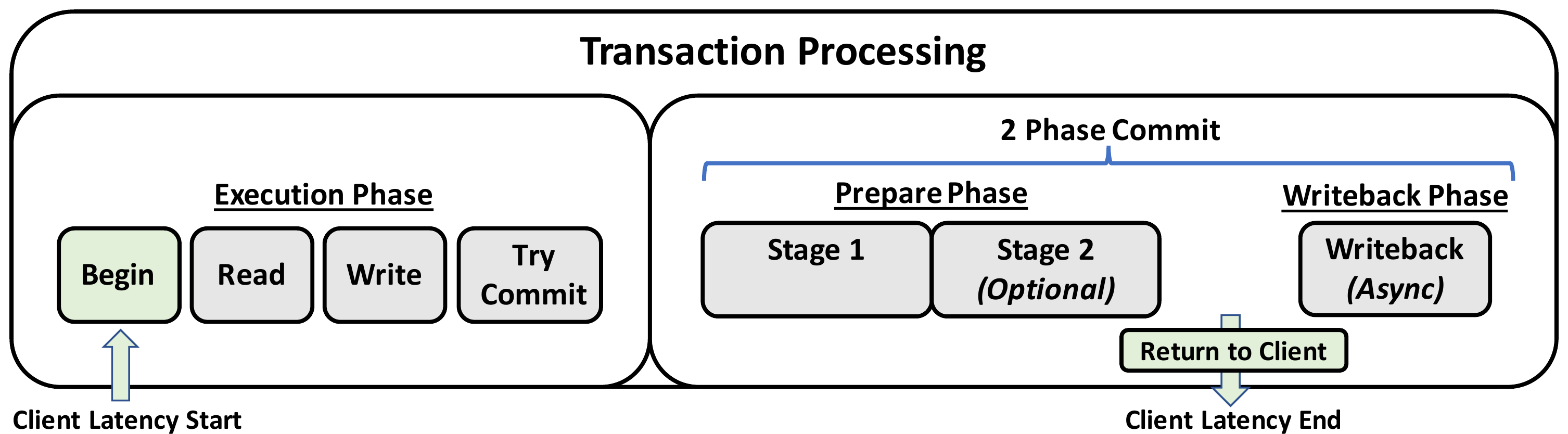}
\caption{Basil Transaction Processing Overview}
\label{fig:overview}
\end{figure}

\sys is a transactional key-value store designed to be scalable and leaderless. Our architecture reflects this ethos. 
\par \textbf{Transaction Processing} Transaction processing is driven by clients (avoiding costly all-to-all communications amongst replicas) and consists of three phases (Figure~\ref{fig:overview}). First, in an \textit{Execution phase}, clients execute individual transactional operations. As is standard in optimistic databases, reads are submitted to remote replicas while writes are buffered locally. \sys{} supports interactive and cross-shard transactions: clients can issue new operations based on the results of
past operations to any shard in the system. In a second \textit{Prepare phase}, individual shards are asked to vote on whether committing the transaction would violate serializability. For performance,  \sys{} allows individual replicas within a shard to process such requests out of order. 
Finally, the client aggregates each shard vote to determine the outcome of the transaction, notifies the application of the final decision, and forwards the decision to the participating replicas in an asynchronous \textit{Writeback phase}.  Importantly, the decision of whether each transaction commits or aborts must be preserved across both benign and Byzantine failures. We describe the protocol in Section~\ref{section:exec}.

\par \textbf{Transaction Recovery} A Byzantine actor could begin executing a transaction, run the prepare phase, but intentionally never reveal its decision. 
Such behavior could prevent other transactions from making progress. \sys{} thus implements a {\em fallback recovery mechanism} (\S\ref{sec:recovery}) that can terminate stalled transactions while preserving Byz-serializability. This protocol, in the common case, allows clients to terminate stalled transactions in a single additional round-trip. 

\par \textbf{Replication} \sys{} uses $n=5f+1$ replicas for each shard. This choice allows \sys to \one preserve Byzantine independence \two commit transactions in a single round-trip in the absence of contention, and \three reduces the message complexity of transaction recovery by a factor of $n$, all features which would not be possible with a lower replication factor. \changebars{We expand on this further in Sections \ref{s:discussion} and \ref{sec:recovery}.}{}

\section{Transaction Processing}
\label{section:exec}
\sys{} takes as its starting point MVTSO~\cite{bernstein1983multiversion}, an aggressive multiversioned concurrency control, and modifies it in three ways: \one~in the spirit of independent operability, it has clients drive the protocol execution; \two~it merges concurrency control with replication; and finally \three~it hardens the protocol against Byzantine attacks to guarantee Byz-serializability while preserving Byzantine independence.

Traditional MVTSO works as follows. A transaction $T$ is assigned (usually by a transaction manager or scheduler) a unique timestamp {\em ts$_T$} that determines its serialization order. As MVTSO is multiversioned, writes in $T$ create new versions of the objects they touch, tagged with {\em ts}$_T$. Reads instead return the version of the read object with the largest timestamp smaller than {\em ts}$_T$ and update that object's  {\em read timestamp} (RTS) to {\em ts}$_T$. Read timestamps are key to preserving serializability: to guarantee that no read will miss a write from a transaction that precedes it in the serialization order, MVTSO aborts all writes to an object from transactions whose timestamp is lower than the object's RTS.

MVTSO is an optimistic protocol, and, as such, much of its performance depends on whether its optimistic assumptions are met. For example, it uses timestamps to assign transactions a serialization order a-priori, under the assumption that those timestamps will not be manipulated; further, it allows read operations to become {\em dependent} on values written by ongoing transactions under the expectation that they will commit. This sunny disposition can make MVTSO particularly susceptible to Byzantine attacks. Byzantine clients could use artificially high timestamps to make lower-timestamped transactions less likely to commit; or they  could simply start transactions that write to large numbers of keys and never commit them: any transaction dependent on those writes would be blocked too.  At the same time, by blocking on dependencies (rather than summarily aborting, as OCC would do) MVTSO leaves open the possibility that blocked transactions may be rescued and brought to commit. In the remainder of this section, we describe how \sys{}, capitalizing on this possibility, modifies MVTSO to harden it against Byzantine faults.

\subsection{Execution Phase}

\par \textbf{Begin()} A client begins a transaction $T$ by optimistically choosing a timestamp \textit{$ts$ $\coloneqq$ (Time, ClientID)} that defines a total serialization order across all clients. Allowing clients to choose their own timestamps removes the need for a centralized scheduler, but makes it possible for  Byzantine clients to create transactions with  arbitrarily high timestamps: objects read by those transactions would cause conflicting transactions with lower timestamps to abort.  To defend against this attack, replicas accept transaction operations if and only if their timestamp is no greater than
$R_{\textit{Time}} + \delta$, where $R_{\textit{Time}}$ is the replica's own local clock. Neither \sys{}'s  safety nor its liveness depend on the specific value of $\delta$, though a well-chosen value will improve the system's throughput. In practice we choose $\delta$ based on the skew of NTP's clock.

\par \textbf{Write(key,value)} Writes from uncommitted transactions raise a dilemma. Making them readable empowers Byzantine clients to stall all transactions that come to depend on them. Waiting to disclose them only when the transaction commits, however, increases the likelihood that concurrent transactions will abort.  We adopt a middle ground: we buffer writes locally until the transaction has finished execution, and make them visible during the protocol's Prepare phase (we call such writes {\em prepared}). This approach allows us to preserve much of the performance benefits of early write disclosure while enforcing independent operability (\S\ref{s:Prepare}).

\par \textbf{Read(key)} In traditional MVTSO, a read for transaction $T$  returns the version of the read object with the highest timestamp smaller than {\em ts}$_T$. When replicas process requests independently, this guarantee no longer holds, as the write with the largest timestamp smaller than {\em ts}$_T$ may have been made visible at replica $R$, but not yet at $R'$: reading from the latter may result in a stale value. Hence, to ensure serializability, transactions in \sys{} go through a \changebars{}{further} concurrency control check at each replica  as part of their Prepare phase (\S~\ref{s:Prepare}). Further care is required, as Byzantine replicas could intentionally return stale (or imaginary!) values that would cause transactions to abort, violating Byzantine independence. These considerations lead us to the following read logic: 
\par \protocol{\textbf{1: C} $\rightarrow$ \textbf{R}: Client C sends read request to replicas.}
C broadcasts an authenticated read request  $m \coloneqq \langle \rd, key, {\mathit ts}_T \rangle$ to \changebars{at least}{$\geq$}  $2f+1$ replicas for shard $S$.
\par \protocol{\textbf{2: R} $\rightarrow$ \textbf{C}: Replica processes client read and replies.}

Each replica $R$ verifies that
the request's timestamp is smaller than $R_{\textit Time} + \delta$.  If not, it
ignores the request; otherwise, it updates {\em key}'s RTS to
{\em ts}$_T$. \sys{} may evict clients with a
history of reading \changebars{keys,}{} but never committing the transaction. Then $R$ returns a signed message \text{$\langle
\textit{Committed,\,Prepared} \rangle_{
\sigma{_R}}$} that contains, respectively, the
latest committed and prepared versions of {\em key} at $R$ 
with timestamps smaller than {\em ts}$_T$.  \textit{Committed} $\equiv$
(\textit{version,\,\ccert}) includes a {\em commit certificate} \ccert
(\S~\ref{s:commit}) proving that {\em version} has committed, while
\textit{Prepared} $\equiv$ ({\em version,\,$id_{T'}$,\,Dep$_{T'}$}) includes a digest identifier for
$T'$ (\S~\ref{s:Prepare}) and the
write-read dependencies {\em Dep}$_{T'}$ of the transaction
$T'$ that created {\em version}. $T'$ cannot commit unless all the transactions
in {\em Dep}$_{T'}$ commit first. 
\par \protocol{\textbf{3: C} $\leftarrow$ \textbf{R}: Client receives read replies.}
A client waits for at least $f+1$  replies (to ensure that at least one comes from a correct replica)  and  chooses the {\em highest-timestamped} version that is {\em valid}. For committed versions, the criterion for validity is straightforward: a committed version must contain a valid \ccert. For prepared versions instead, we require that the same version be returned by at least $f+1$ replicas. Both the validity and timestamp requirement are  important for Byzantine independence. Message validity protects the client's transaction  from becoming dependent on
a version fabricated by Byzantine replicas; and, by choosing the valid reply with the  highest-timestamp, the client is certain to never read a version staler than what it could have read by accessing a single correct replica. 

The client then adds the selected ({\em  key,\,version})  to  {\em ReadSet}$_T$.
If {\em version} was prepared but not committed, it adds a new write-read
dependency to the dependency set {\em Dep}$_T$. Specifically, the client adds to
{\em Dep}$_T$ a tuple ({\em version}, {\em id}$_{T'}$), which will be used during $T$'s Prepare phase to
validate that $T$ is claiming a legitimate dependency.

After $T$ has completed execution, the application tells the client whether it should abort $T$ or instead try to commit it:
\par \textbf{Abort()} The client asks replicas to remove its read timestamps from all keys in {\em ReadSet}$_T$. No actions need to be taken for writes, as \sys{} buffers writes during execution.
\par \textbf{Commit()} The client initiates the Prepare phase, discussed next, which performs the first phase of the multi-shard two-phase commit (2PC) protocol that \sys uses to commit $T$.

\begin{algorithm}
\caption{MVTSO-Check($T$)}\label{mvtso}
\begin{algorithmic}[1]
\If{\textit{$ts_T > localClock + \delta$}} 
\State \Return Vote-Abort
\EndIf

\If{$\exists$ invalid $d \in Dep_T$} 
\State \Return Vote-Abort
\EndIf

\For{\textit{$\forall key,version \in \textit{ReadSet}_T$}}
	\If{$version > ts_T$} \Return MisbehaviorProof 
	\EndIf
        \If{$ \exists T' \in Committed \cup Prepared: key \in \textit{WriteSet}_{T'} $ \newline
        \hspace*{2em} $\land \, version < ts_{T'} < ts_T$}  
          \State  \Return Vote-Abort, \textit{optional: ($T'$, $T'$.\ccert)}  
         \EndIf  
\EndFor

\For{\textit{$\forall key \in \textit{WriteSet}_T$}}
        \If{$\exists T' \in Committed \cup Prepared:$ \newline 
        \hspace*{2em} $\textit{ReadSet}_{T'}\textit{[key].version} < ts_T < ts_{T'}$} 
          \State  \Return Vote-Abort, \textit{optional: ($T'$, $T'$.\ccert)}  
         
        \EndIf
        \If{$\exists RTS \in key.RTS: RTS > ts_T$} 
          \State  \Return Vote-Abort
       \EndIf
\EndFor

\State $Prepared.add(T)$ 

\State \textit{\textbf{wait} for all pending dependencies}
\If{$\exists$ $d \in Dep_T: d.decision= Abort$} 
\State $Prepared.remove(T)$ 
\State \Return Vote-Abort
\EndIf
\State \Return Vote-Commit

\end{algorithmic}
\label{a:MVTSO}
\end{algorithm}

\subsection{Prepare Phase}
\label{s:Prepare}
To preserve independent operability, \sys{} delegates the responsibility for coordinating the 2PC protocol to clients. For a given transaction $T$, the protocols begins with a Prepare phase, which consists of two stages (Figure~\ref{fig:overview}).

In stage \pone, the client collects {\em commit or abort votes} from each shard that $T$ accesses. Determining the vote of a shard in turn requires collecting votes from all the shard's replicas.  To avoid the overhead of coordinating replicas within a shard, \sys{} lets  each replica determine its vote independently, by running a local  \textit{concurrency control check}. The flip side of this design is that, since  transactions may reach replicas  in different orders, even correct replicas within the same
shard may not necessarily reach the same conclusion about $T$. Client $C$ tallies replica votes to learn the vote of each shard and, based on how shards voted, decides whether $T$ will commit or abort. 

Stage \ptwo ensures that $C$'s decision is made durable (or {\em logged}) across failures. $C$ \textit{logs} the evidence on only one shard.  In the absence of contention or failures, \sys's {\em fast path} guarantees that $T$'s decision is already durable and this explicit logging step can be omitted\changebars{, allowing clients to return a commit or abort decision in just a single round trip}{}.

\newpage
\par \textbf{Stage 1: Aggregating votes}
\par \protocol{\textbf{1: C $\rightarrow$ R}: Client sends an authenticated \pone request to all replicas in $S$.}
The message format is $\pone \coloneqq \langle \textsc{prepare}, T \rangle$, where $T$ consists of the transaction's {\em metadata} $\coloneqq$ {\em ts}$_T$, {\em ReadSet}$_T$, {\em WriteSet}$_T$, {\em Dep}$_T$, and of its identifier {\em id}$_T$. To ensure Byzantine clients neither
spoof the list of involved shards nor equivocate $T$'s contents, {\em id}$_T$  is a cryptographic hash of  $T$'s {\em metadata}. 

    \protocol{\textbf{2: R $\leftarrow$ C}: Replica R receives a \pone request and executes the
    concurrency control check.}
Traditional, non-replicated, MVTSO does not require any \changebars{additional validation at commit time,}{validation} as transactions are guaranteed to observe {\em all} the writes that precede them in the serialization order (any "late" write is detected by read timestamps and the corresponding transaction is aborted). This is no longer true in a replicated system: reads could have failed to observe a write performed on a different replica. \sys{} thus runs an additional concurrency control check to determine whether a transaction $T$ should commit and preserve
serializability (Algorithm~\ref{a:MVTSO}). It consists of seven steps: 

\noindent \circled{1} $T$'s timestamp is within the $R$'s time bound (Lines 1-2). \newline
\circled{2} $T$'s dependencies are valid: $R$ has either prepared or committed every  transaction identified by $T$'s dependencies, and the versions that caused the dependencies were produced by said transactions (Lines 3-4).
\newline
\circled{3} Reads in $T$ did not miss any writes. Specifically, the algorithm (Lines 7-8) checks that there does not exist a write 
from a committed or prepared transaction $T'$ that \one~is more recent than the version that $T$'s read and \two~has a timestamp smaller than {\em ts}$_T$ (implying that  $T$ should have observed it). \newline
\circled{4} Writes in $T$ do not cause reads in other {\em prepared
or committed} transactions to miss a write (Lines 9-11). \newline
\circled{5} Writes in $T$ do not cause reads in \textit{ongoing} transactions to miss a write: $T$ is aborted if there exists an RTS greater than {\em ts}$_T$ (Lines 12-13).\newline 
\circled{6} $T$ is prepared and made visible to future reads (Line 14).
\newline
\circled{7} All transactions responsible for $T$'s dependencies have reached a decision. $R$ votes to commit $T$ only if all of its dependencies commit; otherwise it votes to abort (Lines 15-19).

    \protocol{\textbf{3: R $\rightarrow$ C}: Replica returns its vote in a \poner message.}
After executing the concurrency control check, each replica returns to $C$ a Stage1 Reply $\poner \coloneqq \langle T, vote \rangle_{\sigma_R}$. A correct replica executes this check \textbf{at most once} per transaction and stores its vote to answer future duplicate requests (\S \ref{sec:recovery}). 

    \protocol{\textbf{4: C $\leftarrow$ R}: The client receives replicas' votes.}
$C$ waits for \poner messages from the replicas of each shard $S$  touched by $T$. Based on these replies, $C$ determines \one~whether $S$ voted to commit or abort; and \two~whether the received \poner messages constitute a   {\em vote certificate} (\dcert \changebars{$\coloneqq \langle {\mathit id}_T, S, Vote, \{\poner\} \rangle$}{})  that proves  $S$'s vote to be {\em durable}. A shard's vote is durable if  its original outcome can be independently retrieved and verified at any time by any correct client, independent of Byzantine failures or attempts at equivocation. If so, we dub shard $S$ {\em fast}; otherwise, we call it {\em slow}. Votes from a slow shard do not amount to a vote certificate, but simply to a {\em vote tally}\changebars{. Though vote tallies have the same structure as a \dcert, the information they contain is insufficient to make $S$'s vote durable. An additional stage (\ptwo) is necessary to explicitly make S's vote persistent.}{: \sys{} makes them explicitly durable during stage \ptwo of the Prepare Phase.}

Specifically, $C$ proceeds as follows, depending on the set of \poner messages it receives:

\par\noindent\textbf{(1) Commit Slow Path} $(3f+1 \le$ Commit votes $ < 5f+1)$: The client has received at least a  $CommitQuorum\ (CQ)$ of votes, where $\lvert CQ \rvert = \frac{n+f+1}{2} = 3f+1$,  in favor of committing $T$.  Intuitively, the size of {\em CQ} guarantees that two conflicting transactions cannot both commit, since the correct replica that is guaranteed to exist in the overlap of their CQs will enforce isolation. However, $C$ receiving a {\em CQ} of Commit votes is not enough to guarantee that another client
        $C'$, verifying $S$'s vote, would see the same number of Commit votes: after all, $f$ of the replicas in the {\em CQ} could be Byzantine, and provide a different vote if later queried by $C'$.  $C$ thus adds $S$ to the set of slow shards, and records the votes it received in the following {\em vote tally}: $\langle {\mathit id}_T, S, Commit, \{\poner\} \rangle$\changebars{where $\{\poner\}$ is the set of matching (Commit) \poner replies,}{.} \changebars{}{Note that, though a vote tally has the same structure of a \dcert, the information it contains is not sufficient to guarantee that $S$'s vote is durable.}

\par\noindent\textbf{(2) Abort Slow Path} $(f+1 \le$  Abort votes $ < 3f+1)$: A collection of $f+1$ Abort votes constitutes the minimum {\em AbortQuorum (AQ)}, i.e., the minimal evidence sufficient for the client  to count $S$'s vote as Abort in the absence of a conflicting \ccert. Requiring an \textit{AbortQuorum} of at least $f+1$ 
preserves Byzantine independence: Byzantine replicas alone cannot cause a transaction to abort, as \changebars{at least one}{a} correct replica must have found $T$ to be conflicting with a prepared transaction. \changebars{However, such {\em AQ's} are not durable; a client other than $C$ might observe fewer than $f$ abort votes and receive a {\em CQ} instead. $C$ therefore}{The client} records the votes
collected from  $S$ in the following {\em vote tally}: $({\mathit id}_T, S, Abort, \{\poner\}$ and adds $S$ to the slow set for $T$.

\par\noindent \textbf{(3) Commit Fast Path} $(5f+1$ Commit votes$)$: No replica reports a conflict. Furthermore, a unanimous vote ensures that, since  correct replicas never change their vote, \changebars{any}{if some} client $C'$ \changebars{that}{} were to step in for $C$\changebars{}{,  it} would be guaranteed to observe at least a CQ of $3f+1$ Commit votes\changebars{.}{:} \changebars{$C'$ may miss at most $f$ votes because of asynchrony, and at most $f$ more may come from equivocating Byzantine replicas}{}.  $C$ \changebars{thus}{} records the votes collected from  $S$ in the following \dcert: $\langle {\mathit id}_T, S, Commit, \{\poner\} \rangle$ and dubs $S$ fast.

\par\noindent\textbf{ (4) Abort Fast Path} $(3f+1 \le$ Abort votes$)$: $T$ conflicts with a prepared, but potentially not yet committed transaction. $S$'s Abort vote is already durable: since a shard votes to commit only when at least $3f+1$ of its replicas are in favor of it, once $C$ observes $3f+1$ replica votes for Abort from $S$, it is certain that $S$ will never be able to produce $3f+1$ Commit votes, since that would require a correct replica to change its \poner vote or equivocate.  $C$ \changebars{therefore}{then} creates \dcert $\langle {\mathit id}_T, S, Abort, \{\poner\} \rangle$, \changebars{}{where $\{\poner\}$ is the set of matching \poner replies,} and adds $S$ to the set of fast shards.

\par\noindent \textbf{(5) Abort Fast Path} (One Abort with a \ccert for a conflicting transaction $T'$): $C$ validates the integrity of the
\ccert and creates the following \dcert for $S$: $\langle {\mathit id}_T, S, Abort, {\mathit id}_{T'}, \ccert \rangle$. It indicates that $S$ voted to abort $T$ because $T$ conflicts with $T'$, which, as \ccert proves, is a committed transaction. Since \ccert is durable, $C$ knows that the conflict can never be overlooked and that $S$'s vote cannot change; thus, it adds $S$ to the set of fast shards.

After all shards have cast their vote, $C$ decides whether to commit (if all shards voted to commit) or abort (otherwise). Either way, it must make durable the evidence on which its decision is based. As we discussed above, the votes of fast shards  already are; if \changebars{\one there are no slow shards, or \two a single fast shard voted abort}{there are no slow shards (or there is a single fast shard voting to abort)}, then, $C$ can move directly to the Writeback Phase (\S\ref{s:commit}): this is \sys's fast path, which allows $C$ to return a decision for $T$ after a single message round trip. If some shards are in the slow set,  however, $C$ needs to take an additional step to make its tentative 2PC decision durable in a second phase (\ptwo). 
\changebars{Notably though, \sys does \textit{not} need  each slow shard to log its corresponding vote tally in order to make it durable. \fs{ADD REF to Discussion} Instead, \sys first decides whether to commit or abort $T$ based on the shard votes it has received, and then logs its \textit{decision} to only a \textit{single} shard before proceeding to the Writeback phase.}{} 

\par \textbf{Stage 2: Making the decision durable}
\par \protocol{\textbf{5: C $\rightarrow$ R}: The client attempts to make its tentative 2PC decision durable.}
$C$ makes its decision durable by storing an (authenticated) message $\ptwo \coloneqq \langle id_T, decision, \{\vtaly\}, view=0 \rangle$ on {\em one} of the shards that voted in Stage 1 of the Prepare phase; we henceforth refer to this shard, chosen deterministically depending on $T$'s id,  as $S_{\mathit log}$. The set $\{\vtaly\}$ includes the vote tallies of all shards \changebars{to prove the decision's validity. Like many consensus protocols (e.g. \cite{oki1988viewstampeda, castro1999practical, KotlaCACM}), Basil relies on the notion of {\em view} for recovery:}{;} the value of {\em view}  indicates whether this $\ptwo$  message was issued by the client that initated $T$ ({\em view}$=0)$ or it is part of a fallback protocol. We discuss {\em view}'s role in detail in \S\ref{sec:recovery}.

\protocol{\textbf{6: R $\rightarrow$ C}: Replicas in $S_{log}$ receive the $\ptwo$ message and return \ptwor responses.}
Each replica validates that $C$'s 2PC decision is justified by the corresponding vote tallies; if so, the replica logs the decision and acknowledges its success. Specifically, it replies to $C$ with a message of the form \ptwor:$\langle {\mathit id}_T, decision, view_{decision}, view_{current} \rangle_{\sigma_R}$\changebars{; $view_{decision}$ and $view_{current}$ capture additional replica state used during recovery}{}. We once again defer \changebars{an in-depth}{}discussion of \changebars{views}{what a view is} to \S\ref{sec:recovery}. 

    \protocol{\textbf{7: C $\leftarrow$ R}: The client receives a sufficient number of matching replies to confirm a decision was logged.}

$C$ waits for $n-f$  \ptwor messages whose {\em decision} and $view_{decision}$ match, and creates a single shard certificate  $\dcert_{S_{\mathit log}}$: $\langle{\mathit id}_T, S, decision,\{\ptwor \}\rangle$ for the logging shard.

\subsection{Writeback Phase}
\label{s:commit}
$C$ notifies its local application of whether $T$ will commit or  abort, and asynchronously broadcasts to all shards that participated in the Prepare phase a corresponding decision certificate (\ccert for commit; \acert for abort). 

    \protocol{\textbf{1: C $\rightarrow$ R}:
The client asynchronously forwards decision certificates to all participating shards.
    }
$C$ sends to all involved shards a decision certificate (\ccert:$\langle{\mathit id}_T,\,{\mathit Commit},\,\{\dcert_S\}\rangle$ for a Commit decision, \acert: $\langle{\mathit id}_T,\,{\mathit Abort},\,\{\dcert_S\}\rangle$ otherwise). 
We distinguish between the fast, and slow path: On the fast path, \ccert consists of the
full set of Commit \dcert votes from all involved shards, while an \acert need only contain one \dcert vote for Abort. On the slow path, both \ccert and \acert simply include $\dcert_{S_{\mathit log}}$.

    \protocol{\textbf{2: R $\leftarrow$ C}: Replica validates \ccert or \acert and updates store
    accordingly.}
Replicas update all local data structures, including applying writes to the datastore on commit and notifying pending dependencies.

\subsection{An Optimization: Reply Batching}
\label{s:opt}

To amortize the cost of signature generation and verification, \sys{} batches messages \changebars{(Figure~\ref{fig:batching})}{}. Unlike leader-based systems, \sys has no central sequencer through which to batch requests; instead, it implements batching at the replica {\em after} processing messages. To amortize signature generation for replies, \sys{} replicas create batches of $b$ request replies, generate a Merkle tree~\cite{merkle1987digital} for each batch, and sign the root hash. \changebars{They then send to each client $C$ that issued a request: ($i$) the root hash $\mathit root$, ($ii$) a signed version $\sigma$ of the same $\mathit root$, ($iii$) the appropriate request reply $R_C$, and ($iv$) all intermediate nodes (denoted $\pi_C$ in Figure~\ref{fig:batching}) necessary to reconstruct, given $R_C$, the root hash $\mathit root$. Through this batching, the cost of signature generation is reduced by a factor of $b$, at the cost of $log(b)$ additional messages.}{ They then send, to each client, the signed root hash \changebars{(\textit{Merkle commitment} $\coloneqq (root, \sigma)$)}{}, along with its request reply and all intermediate nodes necessary to reconstruct the hash \changebars{(\textit{Merkle proof} $\coloneqq \pi$),}{(} reducing signature generation by a factor of $b$ at the cost of $log(b)$ additional messages\changebars{}{)}.}
To amortize signature verification, \sys{}  uses caching. When a replica successfully verifies the root hash signature in a client message $m$, it caches a map between the corresponding root hash value and the signature. If the replica later receives a message $m'$ carrying  the same  root hash  and  signature as $m$ (indicating that $m$ and $m'$ refer to  the same batch of replies), it can, upon verifying the correctness of the root hash, immediately declare the corresponding signature valid.
\begin{figure}
\begin{center}
\includegraphics[width=0.5\textwidth]{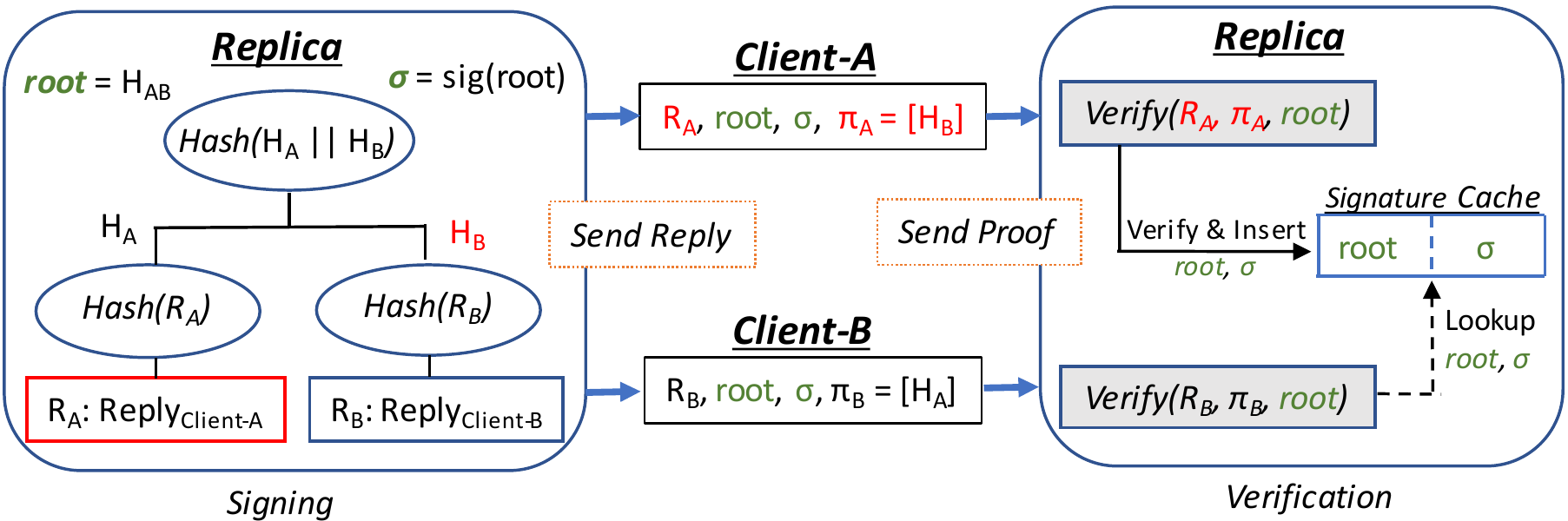}
\end{center}
\caption{\sys batching for two clients. Signature $\sigma$ and batch $root$ (green) are the same across batched replies. Reply $R_C$ and proof $\pi_C$ are unique to each client C and can validate $root$. 
}
\label{fig:batching}
\end{figure}

\subsection{Discussion}
\label{s:discussion}

\noindent{\bf Stripping layers} When 2PC is layered above shards that already order transactions internally using state machine replication, then \textit{within every shard} every correct replica has logged the vote of every other correct replica. 
\sys{}'s design avoids this indiscriminate cost: if all shards are fast, then their votes are already durable without requiring replicas to run any coordination protocol; and if some shards are slow, as we \changebars{mentioned above}{discuss below}, only the replicas of a single shard need to durably log the decision. \changebars{As a result, the overhead of Basil's logging phase remains constant in the number of involved shards.}{}

\changebars{\noindent{\bf Signature Aggregation.}}{} \changebars{\sys, like recent related work \cite{gueta2018sbft, yin2019hotstuff}, can make use of signature aggregation schemes \cite{boneh2018compact, micali2001accountable, boldyreva2003threshold, itakura1983public, shoup2000practical, cachin2005random, gentry2006identity, boneh2003aggregate} to reduce total communication complexity. The client could aggregate the (matching) signed \poner or \ptwor replies into a single signature, thus ensuring that all messages sent by the client remain constant-sized, and hence \sys total communication complexity can be made linear. The current \sys prototype does not implement this optimization. }{}
\fs{Maybe mention the incompatibility with batching... I could also add that in the appendix - I had added a section there about vote subsumption and its incompatibility, and how we can actually get around it. Of course that won't work for batching, but it might be the "right" place to mention}

\changebars{\noindent{\bf Why $n=5f+1$ replicas per shard?} Using fewer replicas has two main drawbacks. First, it eliminates the possibility of a commit fast path. With a smaller replication factor, {\em CQ}s of size $n-2f$ ($f$ can differ because of asynchrony, another $f$ can differ because of equivocation) would no longer be guaranteed to overlap in a correct replica, making it possible for conflicting transactions to commit, in violation of  Byzantine serializability. Second, it precludes Byzantine independence. For progress, clients must  always be able to observe either a {\em CQ} or an {\em AQ}, but, for Byzantine independence, the size of neither quorum must fall below $f$: with $n \le 5f$, it becomes impossible to simultaneously satisfy both requirements.}{}

\section{Transaction Recovery}
\label{section:rec}

For performance, \sys{} optimistically allows transactions to acquire dependencies on uncommitted operations. Without care, Byzantine clients could leverage this optimism to cause transactions issued by correct clients to stall indefinitely. 
To preserve Byzantine independence, transactions must be able to eventually commit even if they conflict with, or acquire dependencies on, stalled Byzantine transactions. To this effect, \sys{} enforces the following invariant: if a transaction acquires a dependency on some other transaction $T$, or is aborted because of a conflict with T,    then a correct participant (client or replica) has enough information to successfully complete $T$. 

Specifically, \sys{} clients whose transactions are blocked or aborted by a stalled transaction $T$ try to finish $T$ by triggering a {\em fallback} protocol. To this end, \sys{} modifies MVTSO to make visible the operations of transactions that have \textit{prepared} only. As $T$'s \pone messages contain all of $T$'s planned writes, any client or replica can use this information to take it upon itself to finish $T$. A correct client is guaranteed to be able to retrieve the \pone  for any of its dependencies, since $f+1$ replicas ({\em i.e.}, at least one correct) must have vouched for that \pone during $T$'s read phase. Likewise, a correct client's transaction only aborts if at least $f+1$ replicas report a conflict. 

\sys{}'s fallback protocol starts with clients: any client blocked by a stalled transaction $T$ can try to finish it. In the \textbf{common case}, it will succeed by simply re-executing the previously described Prepare phase; success is guaranteed as long as replicas within the shard $S_{log}$ that logged shard votes in Stage 2 of $T$'s Prepare phase store the same decision for $T$.  The \textbf{divergent case}, in which they do not, can occur in one of two ways: \one~ a Byzantine client issued $T$ and sent deliberately conflicting \ptwo messages to $S_{log}$; or \two~multiple correct clients tried to finish $T$ concurrently, and collected Prepare phase votes \changebars{(set of \poner messages)}{} that led them to reach (and try to store at $S_{log}$) different decisions. Fortunately,  in \sys{} a Byzantine client cannot generate conflicting \ptwo messages at will: its ability to do so  depends  on the odds (which \S\ref{section:eval}  suggests are low) of receiving, from the replicas of at least one shard, votes that constitute {\em both}  a CQ and an AQ (i.e., {\em 3f+1} Commit votes and {\em f+1} Abort votes). Whatever the cause, if a client trying to finish $T$  observes that replicas in $S_{\mathit log}$ store different decisions, it proceeds to elect a \textit{fallback leader}, chosen deterministically among the replicas in $S_{\mathit log}$. Through this process, \sys{}  guarantees that clients are always able to finish dependent transactions after at most $f+1$ leader elections (since one of them must elect a correct leader).

\label{sec:recovery}
\begin{figure*}
\begin{center}
\includegraphics[width=0.78\textwidth]{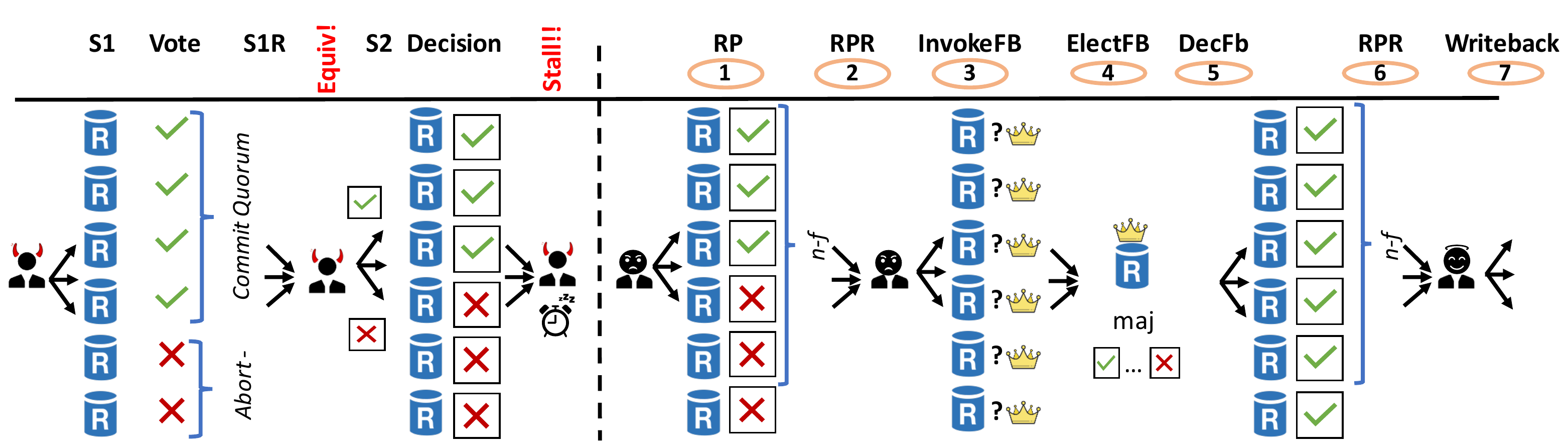}
\end{center}
\caption{Fallback Scenario. A Byzantine client equivocates \ptwor decisions and stalls. An interested client
    invokes the FB}
\label{fig:fallback}
\end{figure*}

Though \sys{}'s fallback protocol is reminiscent of the traditional view-change protocols used to evict faulty leaders, it differs in three significant ways. First, it requires no leader in the common case; further, if electing a fallback leader becomes necessary, communication costs can be made linear in the number of replicas using \changebars{signature aggregation schemes (\S~\ref{s:opt})}{ threshold signatures~cite(gueta2018sbft, shoup2000practical, dan2001short, cachin2005random), as in HotStuff~cite(gueta2018sbft, yin2019hotstuff)}. Second, the fallback election is local, and affects only transactions that access the same operations as the stalled transaction: when a fallback leader is elected for $T$, the scope of its leadership is limited to finishing $T$. In contrast, a standard view-change prevents the system from processing {\em any} operation and the leader, once elected, lords over all consensus operations during its tenure. Finally, \sys' fallback leaders have no say on the ordering of transactions or on what they decide~\cite{zhang2020byzantine}.

As in traditional view-change protocols, each leader operates in a {\em view}. For independent operability, views are defined on a per-transaction basis. Transactions start in $view = 0$;  transactions in that view can be brought to a decision by any client. A replica increases its view number for $T$ each time it votes to elect a new fallback leader.

We now describe the steps of the fallback protocol triggered by a client $C$ wishing to finish a transaction $T$, distinguishing between the aforementioned
common and divergent cases. 

\par \textbf{Common case} In the common case, the client simply resends a \pone message (renamed for clarity \textit{Recovery Prepare} (RP)) in this context) to all the replicas in shards accessed by $T$. Replicas reply with an RPR message which, depending the progress of previous attempts (if any) at completing $T$ corresponds to either \one a \poner message; \two a \ptwor message; or \three a \ccert or \acert certificate. Based on these replies, the client can fast-forward to the corresponding next step in the Prepare or Writeback protocol. In the common case, stalled dependencies thus cause correct clients to experience only a single additional round-trip \changebars{on the fast path, and at most two if logging the decision is necessary (slow path)}{}.

\par \textbf{Divergent case} If, however, the client only receives non-matching \ptwor replies, more complex remedial steps are needed. \ptwor can differ \one in their decision value and \two in their view number $view_{decision}$. The former, as we saw,  is the result  of either explicit Byzantine equivocation or of multiple clients attempting to concurrently terminate $T$. The latter indicates the existence of prior fallback invocations: a Byzantine fallback leader, for instance, could have intentionally left the fallback process hanging. In both scenarios, the client elects a fallback leader. The steps outlined below ensure that, once a correct fallback leader is elected, replicas can be reconciled without introducing live-lock.

    \protocol{\textbf{(1: C $\rightarrow$ R)}: Upon receiving non-matching \ptwor responses, a client starts the fallback process.}
The client sends  $InvokeFB \coloneqq \langle {\mathit id}_T,\,views \rangle$, where {\em views} is the set of signed current views associated with the RPR responses received by the client. 

    \protocol{\textbf{(2: R $\rightarrow$ R$_{\bf FL}$)}: Replicas receive fallback invocation $InvokeFB$ and start election of a fallback leader  R$_{\bf FL}$ for the current view.}
$R$ takes two steps. First, it determines the most up-to-date view $v'$'  held by correct replicas in $S_{log}$ and adopts it as its current view  $view_{current}$.  Second, $R$ sends 
message \fbl : $\langle {\mathit id}_T, decision, view_{current}\rangle_{\sigma_R}$ to the replica with  id $v_{current} + ({\mathit id}_ T \mod n)$ to inform it that $R$ now considers it to be $T$'s fallback leader.

$R$ determines its current view as follows: If a view $v$ appears at least $3f+1$ times in the {\em current views} received in {\em InvokeFB}, then $R$  updates its $view_{current}$  to $max(v+1, view_{current})$; otherwise, it sets its $view_{current}$  to the largest view $v$ greater than its own that appears at least $f+1$ times in {\em current views}. When counting  how frequently a view is present in the received   {\em current views}, $R$ uses vote subsumption: the presence of view $v$ counts as a vote also for all $v' \leq v$.

The thresholds \sys{} adopts  to update a replica's  current view are chosen to ensure that all $4f+1$ correct replicas in $S_{log}$ quickly catch up \changebars{(in case they differ)}{} to the same view,  and thus agree on the identity of the fallback leader. Specifically, by  requiring $3f+1$ matching views to advance to a new view $v$, \sys{} ensures that at least $2f+1$ correct replicas are at most one view behind $v$ at any given time. In turn, this threshold guarantees that \one~a correct client will receive at least $f+1$ matching views for $v' \geq v -1$ in response to its RP message and \two~will include them in its  {\em InvokeFB}.  These $f+1$ matching views are sufficient for all $4f+1$ correct replicas to catch-up to view $v'$, then (if necessary) jointly move to view $v$, and send election messages to the fallback leader of $v$. Refer to \tr{Appendix~\ref{proofs:fb_liveness}, \ref{proofs:subsumption}}{our supplemental material} for additional details and proofs.

    \protocol{\textbf{(3: R$_{\bf FL}$ $\rightarrow$ R)}: Fallback leader R$_{\rm FL}$ aggregates election messages and sends decisions to replicas.}
R$_{\rm FL}$ considers itself elected upon receiving $4f+1$ \fbl messages with matching views $view_{elect}$. It proposes a new decision \textit{$dec_{new}$ $=$ majority}(\{{\em decision}\}) and broadcasts message \fbd:$\langle ( {\mathit id}_T, dec_{new}, view_{elect})_{\sigma_{R_{\rm FL}}},\{\fbl\} \rangle$, which includes the \fbl messages as proof of the sender's leadership.  \changebars{Importantly, an elected leader can only propose \textit{safe} decisions: if a decision has previously been returned to the application or completed the Writeback phase, it must have been \textit{logged} successfully, i.e. signed by $n-f=4f+1$ replicas. Thus, in any set of $4f+1$ \fbl messages the decision must appear at least $2f+1$ times, i.e., a majority. Note that this condition no longer holds when using fewer that $5f+1$ replicas per shard: using a smaller replication factor would require \one an additional (third) round of communication, and \two including proof of this communication in all replica votes (an $O(n)$ increase in complexity), to guarantee that conflicting decision values may not be logged for the same transaction. }{}  

\changebars{}{As in prior work~cite{gueta2018sbft,yin2019hotstuff}, \sys{} could use \changebars{signature aggregation}{threshold signatures} to aggregate \fbl messages into a single signature, keeping the logic linear in both message size and cryptographic costs. }

\par  \protocol{\textbf{(4: R $\rightarrow$ C)}:  Replicas sends a \ptwor message to interested clients.}
Replicas receive a \fbd  message and adopt the message's decision (and $view_{decision}$) as their own if their current view is smaller or equal to $view_{elect}$ \changebars{and the proof is valid.}{.} If so, replicas
update their current view to $view_{elect}$ and forward the decision to all interested clients in a \ptwor message: $\langle {\mathit id}_T, decision, view_{decision}, view_{current} \rangle_{\sigma_R}$.

\begin{figure*}[!t]
    \centering
    \subfloat[Throughput in tx/s]
    {\label{fig:apptput}
    \includegraphics[width=0.46\linewidth,valign=b]
    {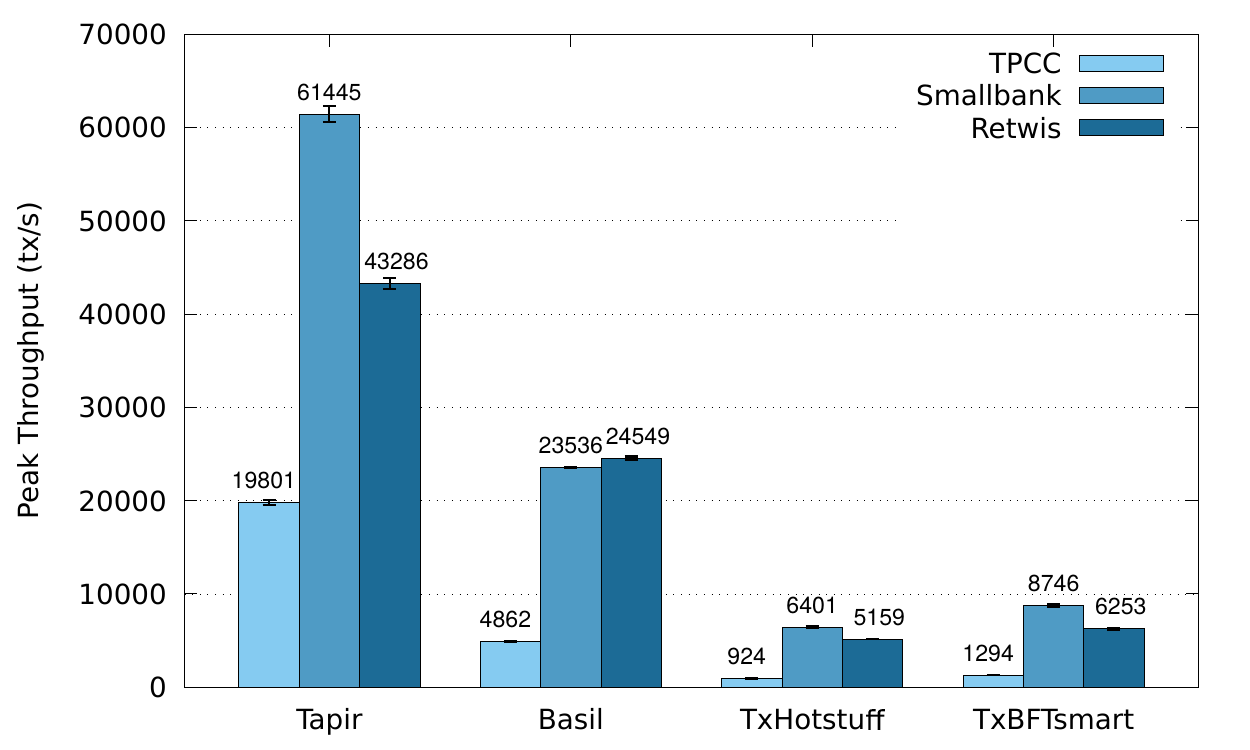}}
    \subfloat[Latency in ms]
    {\label{fig:applat}
    \includegraphics[width=0.46\linewidth,valign=b]
    {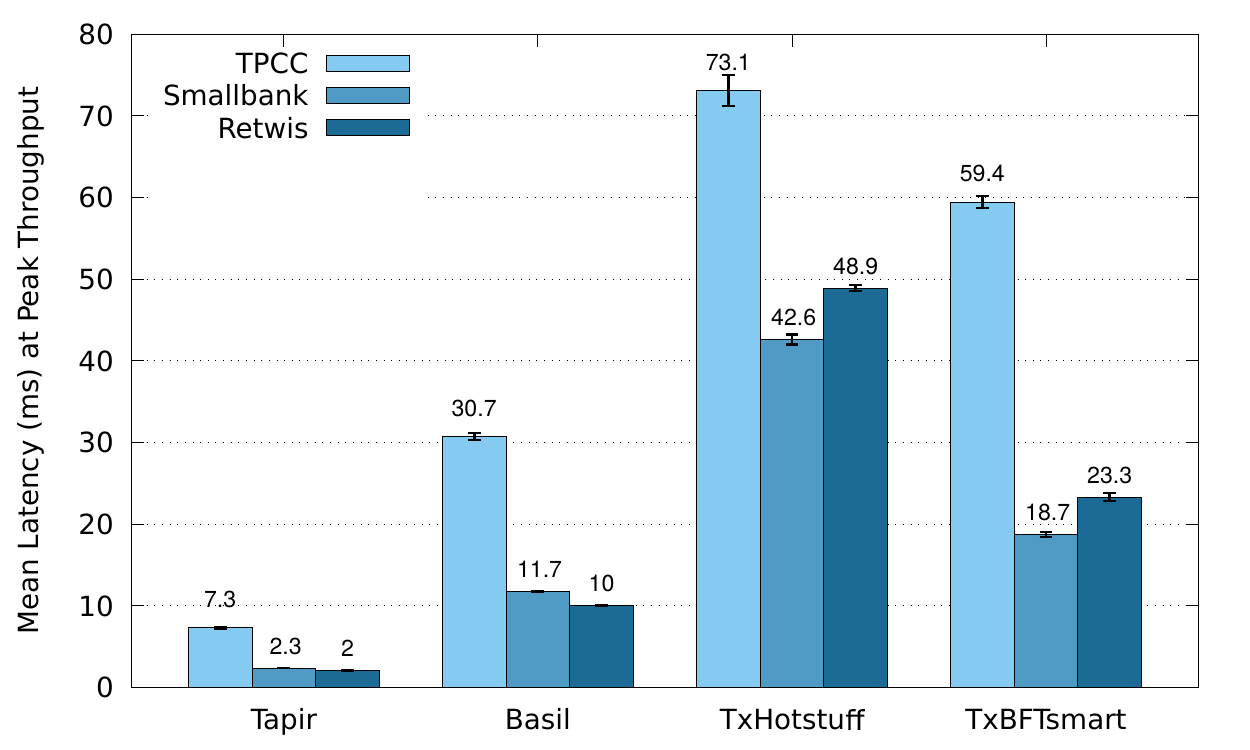}}
    \caption{Application High-level Performance}
\end{figure*}

 \protocol{\textbf{(5: C}: A client creates a \dcert or restarts fallback.}
If the client receives $n-f$ \ptwor with matching decision and decision views, she creates a $\dcert_{S_{\mathit log}}$ and proceeds to the Commit phase.
Otherwise, it restarts the fallback using the newly received $view_{current}$ messages to propose a new view.
$ $\vspace{2mm}

We illustrate the entire divergent case algorithm in Figure~\ref{fig:fallback}, which for simplicity considers a transaction $T$ involving a single shard. \changebars{(With multiple shards, only the common case RP messages (and replies) would involve all shards; the divergent case would always touch only a single shard, $S_{log}$).}{}.

To begin the process of committing  $T$, a Byzantine client sends message \pone and waits for all \poner messages.  Since the  replies it receives allow it to generate both a Commit and Abort quorum, the client chooses to equivocate, sending \ptwo messages for both Commit and Abort. It then stalls. A second correct client who acquired a dependency on $T$ attempts to finish it. It sends \rp messages (1) and receives non-matching \rpr messages (three Commit and two Abort decisions, all from view 0) (2). To redress this inconsistency, the correct client invokes a Fallback with  view 0 (3b). Upon receiving this message, replicas transition to view 1 and send their own decision to view 1's leader in an \fbl message (4). In our example, having received a majority of Commit decisions, the leader  chooses to commit and broadcasts its decision to all other replicas in an \fbd message (5). Finally, replicas send the transaction's outcome to the  interested client (6), who then proceeds to the Writeback phase (7).

\section{Evaluation}
\label{section:eval}

Our evaluation seeks to answer the following questions:
\begin{itemize}
\item How does \sys{} perform on realistic applications? (\S\ref{sec:highlvl})
\item Where do \sys{}'s overheads come from? (\S\ref{sec:overheads})
\item What are the impacts of our optimizations in \sys{}? (\S\ref{sec:opt})
\item How does \sys{} perform under failures? (\S\ref{sec:failures})
\end{itemize}
\vspace{2pt}
\par\textbf{Baselines} We compare against three baselines: \one~TAPIR~\cite{zhang2015tapir},
a recent distributed database that combines replication and cross-shard coordination for greater performance but does not support Byzantine faults; \two~TxHotstuff, a distributed transaction layer built on top of the standard C++ implementation~\cite{libhotstuff} of HotStuff, a recent consensus protocol that forms the basis of several commercial systems \cite{baudet2019state, Celo, Cypherium, Thundercore, meter.io}, most notably Facebook Diem's Libra Blockchain; and  \three~ TxBFT-SMaRt, a distributed
transaction layer built on top of BFT-SMaRt~\cite{bessani2014state, bftsmart}, a state-of-the-art PBFT-based implementation of Byzantine state machine replication~(SMR) \fs{Mention it is java? or do we cut the reference to "standard C++ hotstuff implementation"}. \changebars{We use ed25519 elliptic-curve digital signatures \cite{ed25519, ed25519-donna} for both Basil and the transaction layer of TxHotstuff and TxBFT-SMaRt}{}.
HotStuff and BFT-SMaRt support general-purpose SMR, and are not fully-fledged transactional
systems; we thus supplement their core consensus logic with a coordination layer
for sharding (running 2PC) and an execution layer that implements a standard optimistic concurrency control serializability check~\cite{kung1981occ} and maintains the underlying key-value store. Additionally, we augmented both BFT baselines to also profit from \sys{}' reply batching scheme. 
This architecture follows the standard approach to designing distributed databases (e.g. Google Spanner~\cite{corbett2012spanner}, Hyperledger Fabric~\cite{Hyperledger} or Callinicos \cite{padilha2016callinicos, padilha2013augustus}) where concurrency control and 2PC are layered on top of the
consensus mechanism. Spanner and Hyperledger (built on Paxos and Raft, respectively) are not Byzantine-tolerant, while Callinicos does not support interactive transactions.
\changebars{To the best of our knowledge, this is the first academic evaluation of HotStuff as a component of a large system.}{We dedicate one thread to networking, one to agreement, one to execution, and the remaining threads to serving reads and signature computations in parallel. This split represents the optimal configuration of the baseline.}\footnote{We discussed extensively our setup and implementation with the authors of TAPIR and HotStuff. We corresponded with the authors of Callinicos who were unfortunately unable to locate a fully functional version of their system.}

\par \textbf{Experimental Setup} We use CloudLab~\cite{cloudlab} \texttt{m510} machines (8-core 2.0 GHz CPU, 64 GB RAM, 10 GB NIC, 0.15\textit{ms} ping latency) and run experiments for 90 seconds (30s warm-up/cool-down). \changebars{Clients execute in a closed-loop, reissuing aborted transactions using a standard exponential backoff scheme. We measure the latency of a transaction as the difference between the time the client first invokes a transaction to the time the client is notified that the transaction committed.}{ We run experiments for 90 seconds, following 30 seconds of warm-up and cool-down}  Each system tolerates $f=1$ faults ($n=2f+1$ for TAPIR,  $3f+1$ for HotStuff and BFT-SMaRt).

\subsection{High-level Performance}
\label{sec:highlvl}
We first evaluate \sys{} against three popular benchmark OLTP applications: TPC-C~\cite{tpcc}, Smallbank~\cite{difallah2013oltp}, and the Retwis-based transactional workload used to evaluate TAPIR~\cite{zhang2015tapir}. TPC-C simulates the business logic of an e-commerce framework. We configure it to run with 20 warehouses. As we do not support secondary indices, we create a separate table to \one~ locate a customer's latest order in the \texttt{order status} transaction and \two~ lookup customers by last name in the \texttt{order status} and \texttt{payment} transactions~\cite{crooks2018obladi, su2017tebaldi}. We configure Smallbank, a simple banking application benchmark, with one million accounts. Access is skewed, with 1,000 accounts being accessed 90\% of the time. Users in Retwis, which emulates a simple social network, similarly follow a moderately skewed Zipfian distribution ($0.75$). 
Figures~\ref{fig:apptput} and ~\ref{fig:applat} reports results for the three applications.

\par\noindent{\bf TPC-C} \sys{}'s TPC-C throughput is 5.2x higher than TxHotstuff's and 3.8x higher than TxBFT-SMaRt's -- but 4.1x lower than TAPIR's. All these systems are contention-bottlenecked on the read-write conflict between \texttt{payment} and \texttt{new-order}. \sys{} has 4.2x higher latency than TAPIR: this 
increases the conflict window of contending transactions, and thus the probability of aborts. \sys's higher latency stems from \one its replicas need for signing and verifying signatures; \two its larger quorum sizes for both read and prepare phases; and \three its need to validate read/prepare replies at clients.

Throughput in \sys{} is higher than in TxHotStuff and TxBFT-SMaRt.  \sys's superior performance is directly linked to its lower latency (2.4x lower than TxHotstuff, 1.2x lower than TxBFT-SMaRt).
By merging 2PC with agreement, \sys{} allows transactions to decide \changebars{whether to commit or abort}{} in a single round-trip 96\% of the time (through its fast path) \changebars{}{to commit or abort}. TxHotstuff and TxBFT-SMaRt, which  layer a 2PC protocol over a black-box consensus instance, must instead process and order two requests for each decision (one to Prepare, and one to Commit/Abort), each requiring multiple roundtrips.  In particular, Hotstuff and BFT-SMaRT incur respectively nine and five message delays before returning the Prepare result to clients.  In a contention-heavy application like TPC-C, this higher latency translates directly into lower throughput, since it significantly increases the chances that transactions will conflict. Indeed,  for these applications layering transaction processing on top of state machine replication actually turns a classic performance booster for state machine replication---running agreement on large batches---into a liability, as large batches increase latency and encourage clients to operate in lock-step,
increasing contention artificially. In practice, we find that TxHotstuff  and TxBFT-SMaRt perform best with comparatively small batches (four transactions for TxHotStuff, 16 for TxBFT-SMaRT).

\par\noindent{\bf Smallbank and Retwis} \sys{}  is only 1.8/2.6x slower than TAPIR for these workloads, which are resource bottlenecked for both systems.
The lower contention in Smallbank and Retwis (due to the relatively small transactions) allows \sys{} to use a batch size of 16 for signature generation (up from 4 in TPC-C), thus lowering the cryptographic overhead that \sys{} pays over TAPIR. With this larger batch, both TAPIR and \sys{} are bottlenecked on message serialization/deserialization and networking overheads. Because of their higher latency, however, TxHotStuff and TxBFT-SMaRt continue to be contention bottlenecked: \sys's commit rates for Smallbank and Retwis are respectively  93\% and 98\%, but for TxHotStuff they drop to 75\% and 85\% and for TxBFT-SMaRt to 85\% for both benchmarks. Even on their best configuration (batch size of 16 for TxHotStuff and 64 for TxBFT-SMaRt), \sys outperforms them, respectively, by 3.7x and 2.7x  on Smallbank,  and by 4.8x and 3.9x on Retwis.  

\begin{figure*}[!th]
    \centering
    \subfloat[Impact of signatures]
    {\label{fig:crypto}
    \includegraphics[width=0.30\linewidth,valign=b]
    {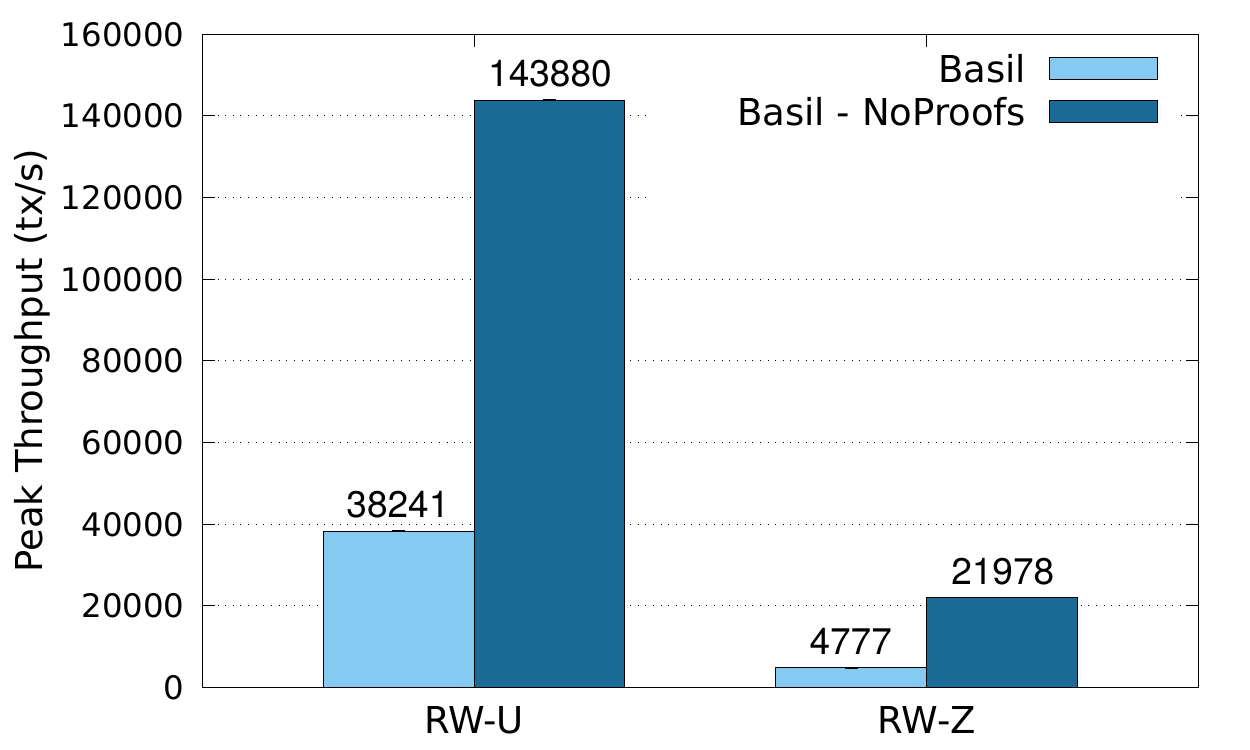}}
    \subfloat[Impact of read quorums]
    {\label{fig:reads}
    \includegraphics[width=0.30\linewidth,valign=b]
    {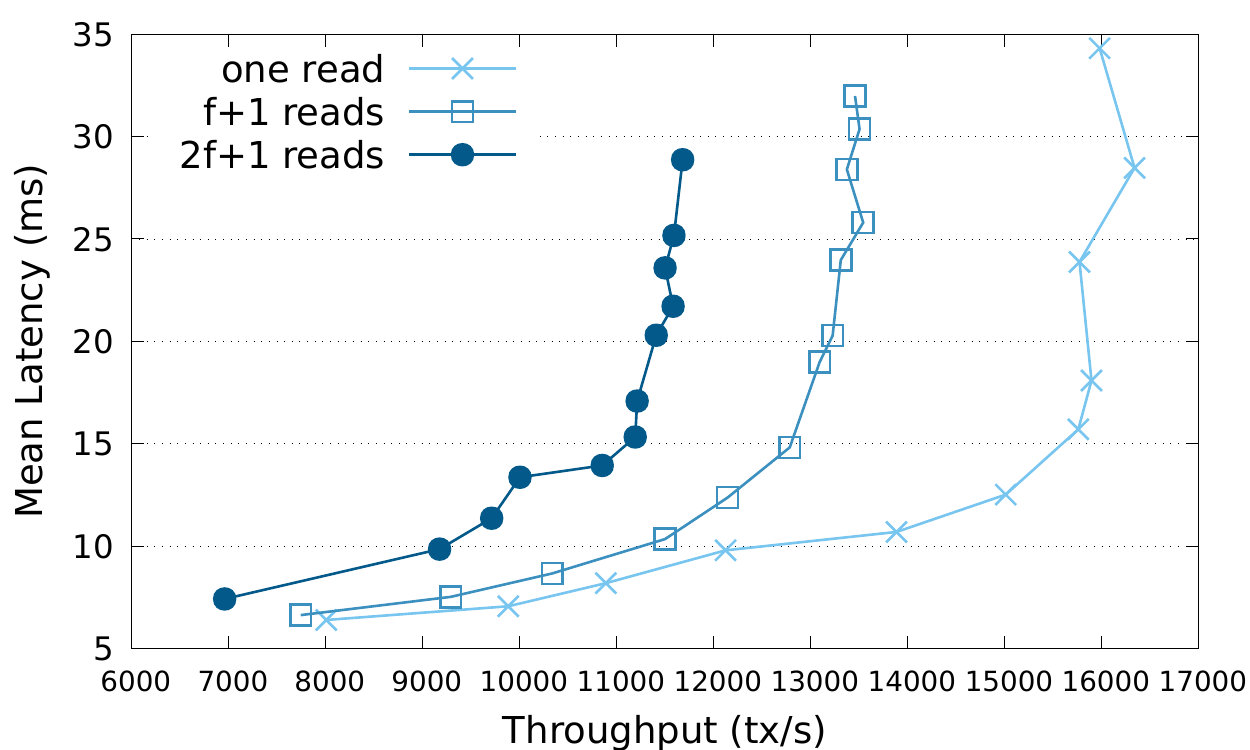}}
    \subfloat[Impact of shard count]
    {\label{fig:shards}
    \includegraphics[width=0.30\linewidth,valign=b]
    {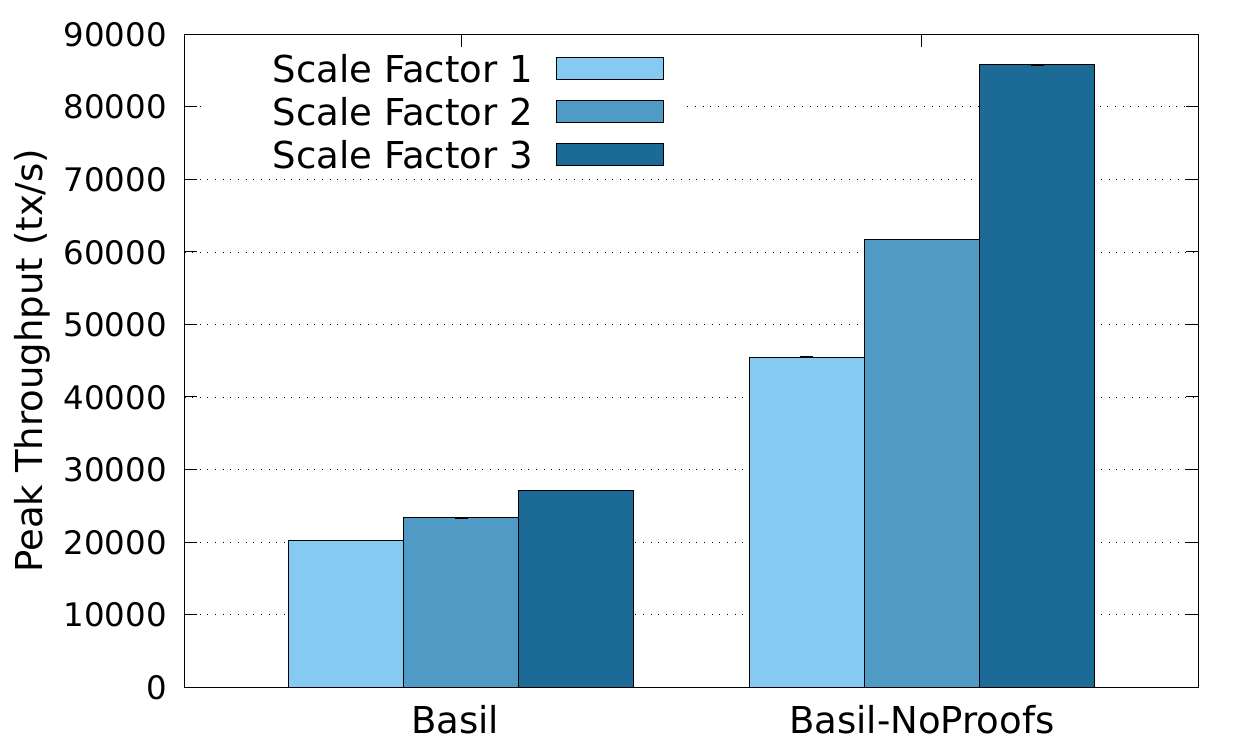}}
    \caption{BFT Overheads}
\end{figure*}

\subsection{BFT Overheads}
\label{sec:overheads}
Besides additional replicas (from $2f+1$ to $5f+1$), tolerance
to Byzantine faults requires both additional cryptography to preserve
Byz-serializability and more expensive reads to preserve Byzantine
independence. To evaluate these overheads, we configure the YCSB-T microbenchmark
suite~\cite{Cooper}  to implement a simple workload of identical
transactions over ten million keys. We distinguish
between a uniform workload (\textit{RW-U}) and a Zipfian workload
(\textit{RW-Z}) with coefficient $0.9$.

We first quantify the cost of cryptography. To do so, we measure the relative
throughput of \sys{} with and without signatures. Transactions consist of 
two reads and two writes. \changebars{}{We use batches of size
16 and 4 respectively for \sys{} on both workloads.} Figure~\ref{fig:crypto} describes our results. We find that \sys{}
without cryptography performs 3.7x better than \sys{} with cryptography
on the uniform workload, and up to 4.6x better on the skewed workload.
Without cryptography, \sys{} can use cores that would have been dedicated for signing/
signature verification for regular operation processing. This effect is more pronounced
on the skewed workload as reducing latency (through increased operation parallelism, lack of batching, and absence of signing/verification latency)
reduces contention, and thus further increases throughput.

In all sharded BFT systems, the number of signatures necessary per transaction grows linearly with 
the number of shards: each replica must verify that other shards also voted to commit/abort a transaction
before finalizing the transaction decision locally. This requires a signature per shard. In Figure~\ref{fig:shards},
we quantify this cost by increasing the number of shards from one to three on the CPU-bottlenecked \textit{RW-U}
workload (three reads/writes). \changebars{}{Each shard is located on a different machine,
and we run the system with a batch size of 16. Additionally, we increase the number of reads and writes to three each.}
\sys{} without cryptography increases by a factor of 1.9 (on average, transactions with three read operations will touch
two distinct shards). In contrast, \sys{}'s throughput increases by only 1.3x.

To guarantee Byzantine independence, individual clients must receive responses from $f+1$ replicas instead of a single
replica. Reading from $2f+1$ replicas (thus sending to 3f+1) increases the chances of a transaction acquiring a valid dependency over reading outdated
data. We measure the relative cost of these different read quorum sizes in Figure~\ref{fig:reads}. We use a simple
read-only workload of 24 operations per transaction, and a batch size of 16. Unsurprisingly, increasing the number of read operations increases the load on each replica
due to the \one~additional signature generations that must be performed, and \two~the additional messages that must be processed. Throughput
decreases by 20\% when reading from $f+1$ replicas (instead of one), and a further 16\% when reading from $2f+1$.

\begin{figure}[H]
\vspace*{-3mm}
    \centering
    \subfloat[Throughput with/without fast path]
    {\label{fig:fast}
    \includegraphics[width=0.5\linewidth,valign=b]
    {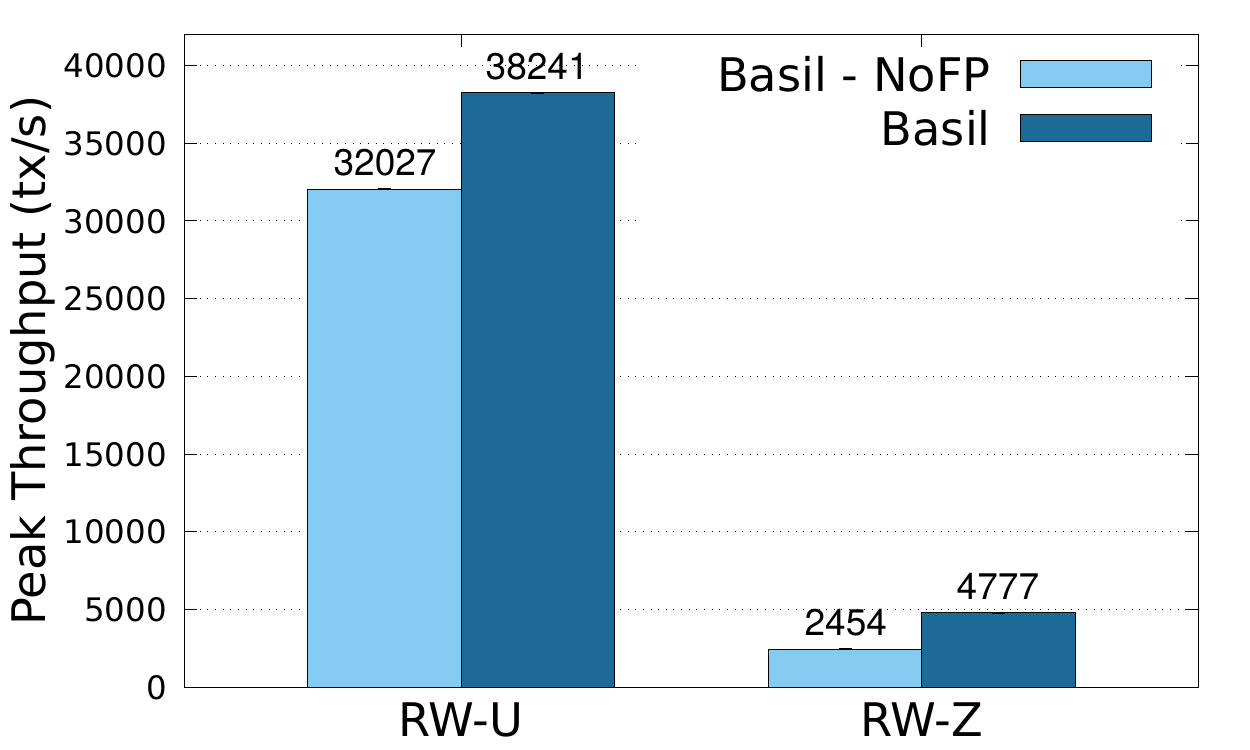}}
    \subfloat[Throughput vs. batch size]
    {\label{fig:batches}
    \includegraphics[width=0.5\linewidth,valign=b]
    {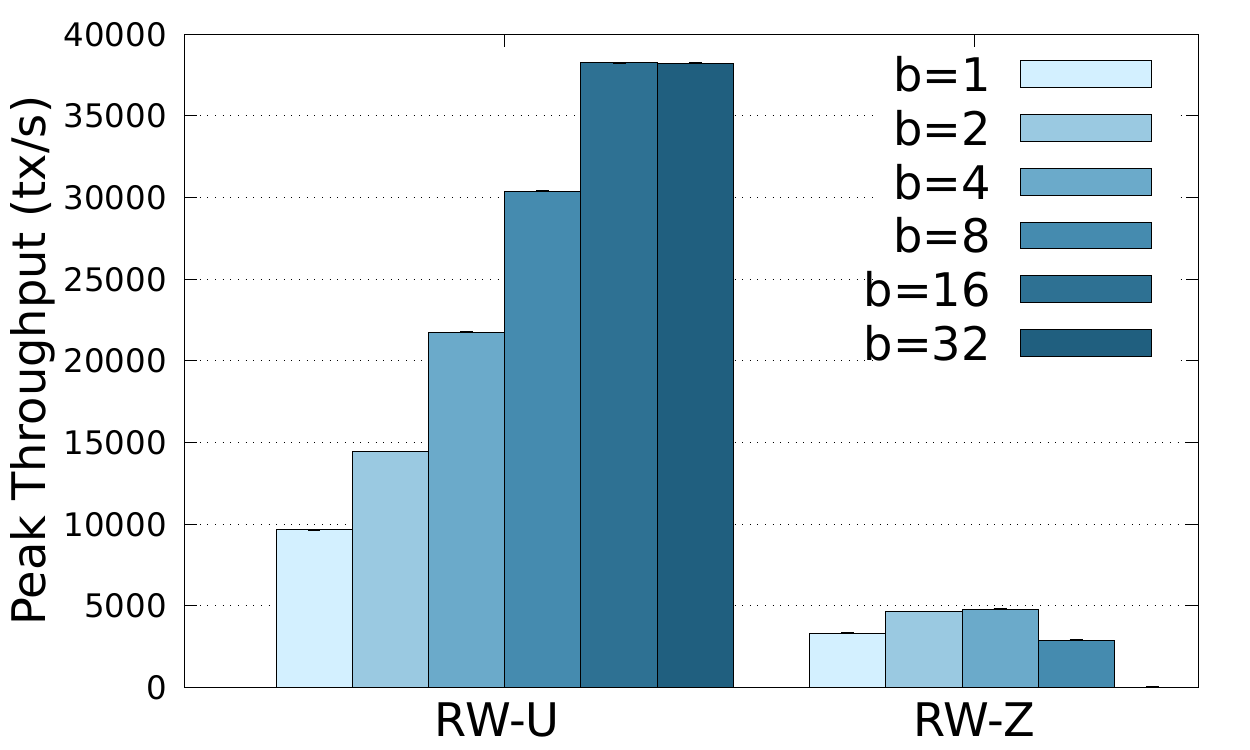}}
    \caption{\sys{} Optimizations}
    \vspace*{-3mm}
  \end{figure}

\subsection{\sys{} Optimizations}
\label{sec:opt}

We measure how \sys's performance benefits from  batching and from its fast path option. We report results for YCBS-T with and without fast path (NoFP) on a workload of two reads and two writes (Figure~\ref{fig:fast}). For the uniform workload, enabling fast paths leads to a 19\% performance increase;  the \ptwor messages that fast paths save contain a signature that must be verified, but require little additional processing. For  a contended Zipfian workload, however, the additional phase incurred by the  slow path  increases contention (as it increases latency): adding the fast path increases throughput by 49\%. Note that Byzantine replicas, by refusing to vote or voting abort, can effectively disable the fast path option; \sys{} can prevent this by removing consistently uncooperative replicas. \fs{In a permissioned system, slow replicas may be treated as faulty (and hence be replaced) since they are harming the system performance.}

\begin{figure*}[!t]
    \centering
    \subfloat[Tput vs. Failures (RW-U)]
    {\label{fig:failures_u}
    \includegraphics[width=0.46\linewidth,valign=b]
    {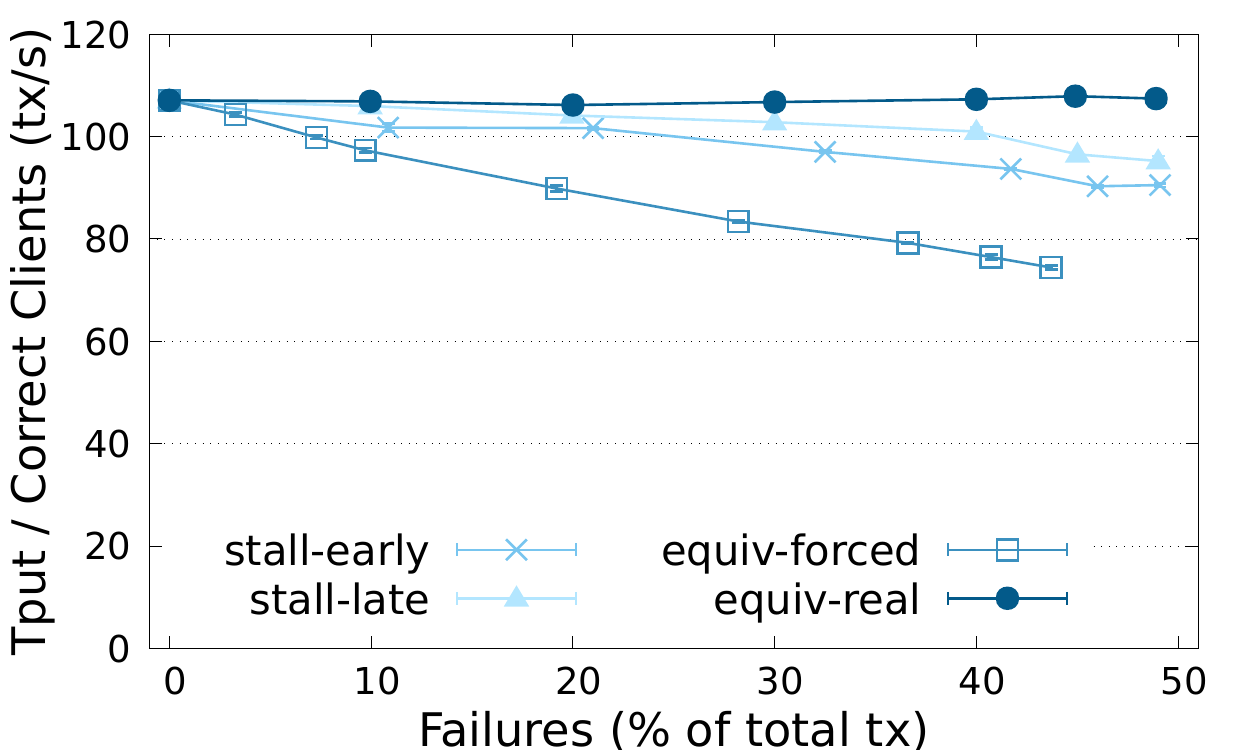}}
    \subfloat[Tput vs. Failures (RW-Z)]
    {\label{fig:failures_z}
    \includegraphics[width=0.46\linewidth,valign=b]
    {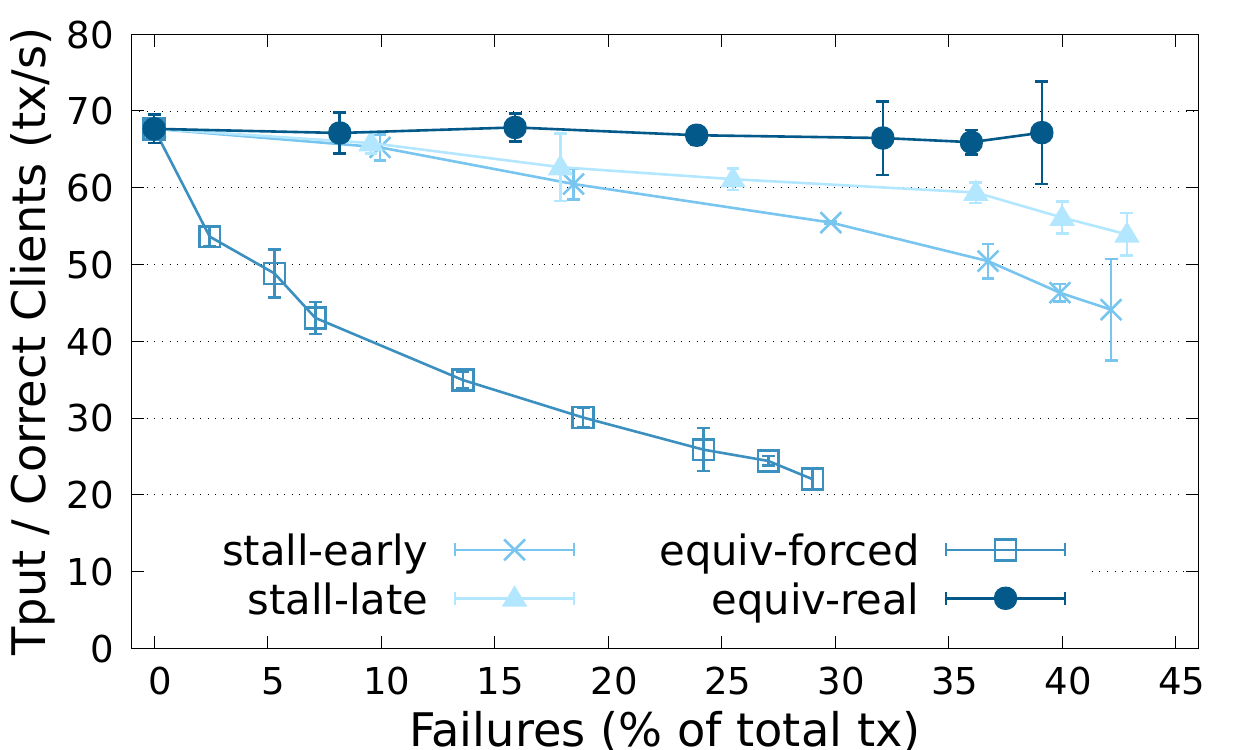}}
    \caption{\sys{} performance under Client Failures}
  \end{figure*}

Next, we quantify the effects of batching. We report the throughput for both workloads (transactions consist of two reads and two writes) while changing the batch size from 1 to 32 ((Figure~\ref{fig:batches}). As expected, on the resource-bottlenecked uniform workload, throughput increases linearly with increased batch size until peaking at 16 (a 4x throughput increase) -- at which point additional hashing costs of the batching Merkle tree nullify any further reduction in signature costs. On the Zipfian workload instead, throughput only increases by up 1.4x, peaking at a small batch size of 4, and degrading afterwards as higher wait times and batch-induced client lock step increase contention (thus reducing throughput).

\subsection{\sys{} Under Failures}
\label{sec:failures}
\sys{} can experience Byzantine failures from both replicas and clients. We have already quantified the effect of Byzantine replicas preventing fast paths (Figure~\ref{fig:fast}) and, by  being unresponsive, forcing ($2f+1$)-sized read quorums (Figure~\ref{fig:reads}). We then focus here on quantifying the effects of clients failing.

\changebars{Basil clients are in charge of their own transactions. Byzantine clients can thus only disrupt honest participants when their own transactions conflict with those of correct clients. Otherwise, they hurt only themselves. A Byzantine client's best strategy for successfully disrupting execution is
\one to follow the (estimated/observed) workload access distribution (only contending transactions cause conflicts), \two choose conservative timestamps (only committing transactions cause conflicts) and \three delay committing a prepared transaction (forcing dependencies to block, and conflicts to abort).}{}

Byzantine clients can stall after sending \pone messages (\textit{stall-early}), or before sending vote certificates $\dcert_S$ (\textit{stall-late}). To equivocate, they must instead receive votes that allow them to  generate, and send to replicas,  conflicting $\dcert$ certificates. \changebars{We remark that equivocating, and hence triggering the divergent case recovery path,  is \textit{not} a strategy that can be pursued deterministically or even just reliably, since its effectiveness depends on the luck of the draw in the composition  of \poner's quorum.}{}
We evaluate two scenarios: a worst-case, in which we artificially allow clients to always equivocate (\textit{equiv-forced}), and a realistic setup where clients only equivocate when the set of messages received allows them to (\textit{equiv-real}). 
For both scenarios, we report  the throughput of correct clients (measured in $tx/s/client_{correct}$). We keep a constant number of clients, a fraction of which exhibits Byzantine behavior in some percentage \changebars{of their newly admitted}{} transactions; we refer to those transactions as faulty). \changebars{Faulty transactions that abort because of  contention are not retried, while correct transactions that abort for the same reason may need to re-execute (and hence prepare) multiple
times until they commit. When measuring throughput, we report the percentage of faulty transactions as a fraction of all processed (not admitted) transactions. }{} Figures~\ref{fig:failures_u} and~\ref{fig:failures_z} illustrate our results.

For the RW-U workload, the additional CPU cost of fallback invocations on the CPU-bottlenecked servers causes correct clients' throughput to decrease slowly and linearly. Clients invoke fallbacks only rarely, as there is no contention. Moreover, stalled transactions can be finished in a single round-trip (a pair of \rp, \rpr) messages thanks to the fallback's common case and fast path. The small throughput
drop over \textit{stall-late} is an artefact of Byzantine clients directly starting a new transaction before finishing the old one, increasing the throughput of malicious clients over correct ones.  The cost of 
forced equivocation is higher as it requires three rounds of message processing (fallback invocation, election, and decision adoption). In reality,
\textit{equiv-real} sees no throughput drop, as the lack of contention makes equivocation impossible: Byzantine clients cannot build the necessary  conflicting \dcert's.

The RW-Z workload is instead contention-bottlenecked: higher latency implies more conflicts, and thus lower throughput. The impact of \textit{stall-late} stalls remains small, as all affected clients still recover the transaction on the common case fast path (incurring only one extra roundtrip).  The performance degradation is slightly higher in \textit{stall-early}, as stalled-early transactions do not finish the transactions on which they depend before stalling. Instead, affected correct clients must themselves invoke the fallback for stalled dependent transactions, which increases latency.  In practice, dependency chains remain small: because of the Zipfian nature of the workload, correct clients quickly notice stalled transactions and aggressively finish them. We note that stalled transactions do not themselves increase contention: \sys allows
the stalled  writes of prepared but uncommitted transactions to become visible to other clients as dependencies. A stalled transaction thus causes dependency chains to grow, but does not increase the conflict window. The throughput drop that results from forcing equivocation failures is, in contrast, significant: equivocation requires three round-trip to resolve and may lead to transactions aborting and to cascading aborts in dependency chains. In practice, Byzantine clients in
\textit{equiv-real} are
once again rarely successful in obtaining the
conflicting \poner messages necessary to equivocate, even in a contended workload (0.048\% of the time for 40 \% faulty transactions) as 99\% of transactions commit or abort on the fast path.

\changebars{Basil expects some level of cooperation from its participants and can remove, without prejudice, clients that frequently stall or timeout (in addition to explicitly misbehaving clients). To avoid spurious accusations towards correct clients, such exclusion policies can be lenient, since Basil's throughput remains robust even with high failure rates.}{}

\section{Related Work}  
\label{section:rel}
State machine replication (SMR) \cite{schneider1990implementing} maintains a total order of requests across replicas, both in the crash failure model \cite{Li2007, Lampson2001, Lamport98Paxos, Lamport2005a, Lamport2005, Lamport01Paxos, Chandra2007, junqueira2011zab, van2014vive, oki1988viewstampeda, liskov2012viewstamped, ongaro2014search}) and in the Byzantine setting \cite{castro1999practical, martin2006fast, KotlaCACM,  gueta2018sbft,clement09aardvark,buchman2016tendermint, yin2019hotstuff, Clement09Upright, pires2018generalized, bessani2014state, lamport2011byzantizing, arun2019ezbft, malkhi2019flexible, duan2014hbft, yin2003separating, Guerraoui08Next, Kotla04High,  liskov2010viewstamped}, where they  have been used as a main building block of Blockchain systems \cite{androulaki2018hyperledger, Hyperledger, EthereumQuorum, buchman2016tendermint, al2017chainspace, kokoris2018omniledger,
gilad2017algorand, baudet2019state}. To maintain a total order abstraction, existing systems process  all operations sequentially (for both agreement and execution), thus limiting scalability for commutative workloads. 
They are, in addition, primarily leader-based which introduces additional scalability bottlenecks \cite{moraru2013there, zhang2015tapir, stathakopoulou2019mir} as well as fairness concerns. Rotating leaders~\cite{clement09aardvark, buchman2016tendermint,
yin2019hotstuff} reduce fairness concerns, and multiple-leader based systems \cite{moraru2013there, stathakopoulou2019mir, arun2019ezbft, li2016sarek} increase throughput. Recent work \cite{zhang2020byzantine, kelkar2020order, kursawe2020wendy, herlihy2016enhancing} discusses  how to improve fairness in BFT leader-based systems with supplementary ordering layers and censorship
resilience. \sys{} sidesteps these concerns by adopting a leaderless approach and addresses the effects of Byzantine actors beyond ordering through the stronger notion of Byzantine Independence. 

\par \textbf{Fine-grained ordering} Existing replicated systems in the crash-failure model leverage operation semantics to allow commutative operations to execute concurrently~\cite{moraru2013there, Lamport2005, sutra2011fast, li2012redblue, park2019exploiting, yan2018carousel, mu2016consolidating, mu2014extracting, zhang2015tapir, kraska2013mdcc}. This work is much rarer
in the BFT context, with Byblos \cite{bazzi2018clairvoyant} and Zzyzyx \cite{hendricks2010zzyzx} being the only BFT protocol that seek to leverage commutativity. However, unlike \sys{}, Byblos is limited to a static transaction model and introduces
blocking between transactions that are {\em potentially} concurrent with other conflicting transactions; while Zzyzyx resorts to a SMR substrate protocol under contention. Other existing Quorum-based systems naturally allow for non-conflicting operations to execute concurrently, but do not provide transactions~\cite{malkhi1998Byzantine, abd2005fault, cowling2006hq, liskov2006tolerating}.

\par \textbf{Sharding} Some Blockchains rely on sharding to parallelize independent transactions, but
continue to rely on a total-order primitive within shards~\cite{zamani2018rapidchain, al2017chainspace, kokoris2018omniledger}. As others in the crash-failure model have highlighted~\cite{zhang2016operation,zhang2015tapir,mu2016consolidating}, this approach incurs redundant coordination and fails to fully leverage the available parallelism within a workload. 
\par\textbf{DAGs} Other permissionless Blockchains use directed acyclic graphs
rather than chains~\cite{pervez2018comparative, popov2016tangle, rocket2018snowflake}, but require dependencies and conflicts to be known prior to execution.
\par \textbf{Byzantine Databases} \sys{} argues that BFT systems and Blockchains are in fact simply
databases and draws on prior work in BFT databases. HRDB \cite{vandiver2007tolerating} offers interactive transactions for a replicated database, but relies on a trusted coordination layer. Byzantium \cite{garcia2011efficient} designs a middleware system that utilizes PBFT as atomic broadcast (AB) and provides Snapshot Isolation using a primary backup validation scheme. Augustus \cite{padilha2013augustus} leverages sharding for scalability in the mini-transaction model \cite{aguilera2007sinfonia} and relies on AB to implement an optimistic locking based execution model. Callinicos \cite{padilha2016callinicos} extends Augustus to support armored-transactions in a multi-round AB protocol that re-orders conflicts for robustness against contention.
 BFT-DUR \cite{pedone2012Byzantine} builds interactive transactions atop AB, but does not allow for sharding. \changebars{}{; Our TxHotstuff baseline resembles the spirit of BFT-DUR while extending to sharding.}
 \sys{} instead supports general transactions and sharding without a leader or the redundant coordination introduced by atomic broadcast.
\par \textbf{Byzantine Clients} \sys{}, being client-driven, must  defend against
Byzantine clients. It draws from prior work targeted at reducing the severity and frequency of client misbehavior~\cite{liskov2006tolerating, garcia2011efficient, pedone2012Byzantine, padilha2013augustus, padilha2016callinicos, luiz2011Byzantine} and extends Liskov and Rodrigues' \cite{liskov2006tolerating}
definition of Byz-Linearizability to formalize the first safety and liveness properties for transactional BFT systems.

\section{Conclusion}
\label{section:conc}
\fs{can add back the two commented sentences we cut for space?}
This paper presents \sys, the first leaderless BFT transactional key-value store supporting ACID transactions.
\sys offers the \textit{abstraction} of a totally-ordered ledger while supporting highly concurrent transaction processing and ensuring \textit{Byz-serializability}. \sys clients make progress \textit{independently}, while \textit{Byzantine Independence} limits the influence of Byzantine participants. 
During fault and contention-free executions \sys commits transactions in a single round-trip.

\section*{Acknowledgments}
\changebars{We are grateful to Andreas Haeberlen (our shepherd), Malte Schwarzkopf, and the
anonymous SOSP reviewers for their thorough and insightful
comments. Thanks to Daniel Weber and Haotian Shen for their help in developing early
versions of our baseline systems. This work was supported in part by
the NSF grants CSR-17620155, CNS-CORE 2008667, CNS-CORE 2106954, and by a Facebook
Graduate fellowship.}{}
\bibliographystyle{abbrv}
\newpage
\bibliography{References/refs}
\onecolumn\twocolumn
\tr{
\newpage
\onecolumn\twocolumn
\appendix

\section{Correctness Sketch}
\label{sec:correctness}
We sketch the main theorems and lemmas necessary to show safety and liveness for \sys{}.
Full proofs can be found in Appendix~\ref{sec:proofs}
We proceed in two steps: we first prove safety without the fallback protocol. We then
extend our correctness argument to handle fallback cases. 

First, we prove that each replica generates a locally serializable schedule.

\par \textbf{Lemma} \ref{proof:mvtso}. {\em  On each correct replica, the set of transactions for which the MVTSO-Check returns Commit 
forms an acyclic serialization graph.}

We then show that decisions for transactions are unique.
\par \textbf{Lemma} \ref{proof:logged}. {\em There cannot exist both an \ccert and a \acert for a given transaction.}

Next, we define a notion of conflicting transactions: a transaction $T_i$ conflicts
with $T_j$ if adding $T_j$ to a history containing $T_i$ would cause the execution to violate
serializability.
\par \textbf{Lemma} \ref{proof:visible}. {\em  If $T_i$ has issued a \ccert and $T_j$ conflicts with $T_i$, then $T_j$ cannot issue a \ccert.}
    
Given these three lemmas, we prove that \sys satisfies Byz-Serializability.
\par \textbf{Theorem} \ref{proof:ser}. {\em \sys{} maintains \textit{Byzantine-serializability}.}

Finally, we show that \sys preserves Byzantine independence under non-adversarial network assumptions. 

\par \textbf{Theorem} \ref{proof:independence}. {\em \sys{} maintains Byzantine independence in the absence of network adversary.}

We now explicitly consider the fallback protocol.  We first show that certified decisions remain durable.
\par \textbf{Lemma} \ref{proof:durable}. {\em Fallback leaders cannot propose decisions that contradict existing slow path \ccert/\acert. }

Additionally, we show that any reconciled decision must have been proposed by a client.
\par \textbf{Lemma} \ref{proof:validity}. {\em Any decision proposed by a fallback leader was proposed by a client.}
 
Given these two Lemmas we conclude:
\par \textbf{Theorem} \ref{proof:safety}. {\em Invoking the fallback mechanism preserves Theorem~\ref{proof:ser} and Theorem~\ref{proof:independence}.}

Next we show that, given partial synchrony~\cite{dwork1988consensus} and the existence of a global stabilization time (GST), \sys's fallback mechanism guarantees progress for correct clients.
\par \textbf{Theorem} \ref{proof:bounds}. A correct client can reconcile correct replicas' views in at most two round-trips and one time-out.

We use this result to prove that:
\par \textbf{Lemma} \ref{proof:correct_leader}. {\em After GST, and in the presence of correct interested clients, a correct fallback leader is eventually elected.} 

Given this lemma, we show that \sys{} allows correct clients to complete their dependencies during synchronous periods. 
\par \textbf{Theorem} \ref{proof:live}. {\em  A correct client eventually succeeds in acquiring either a \ccert or \acert for any transaction of interest.}

\section{Proofs}
\label{sec:proofs}

\subsection{Byzantine-Serializability}

\begin{lemma}
On each correct replica, the set of transactions for which the MVTSO-Check returns Commit 
forms an acyclic serialization graph.
\label{proof:mvtso}
\end{lemma}

\begin{proof}
  We use Adya's formalism here~\cite{adya99weakconsis}.
 An execution of \sys{} produces a direct serialization graph (DSG) whose
  vertices are committed transactions, denoted $T_t$, where $t$ is the unique timestamp identifier.
  Edges in the DSG are one of three types:

    \begin{itemize}
    \item $T_i \xrightarrow{ww} T_j$ if $T_i$ writes the version of object $x$ that precedes $T_j$ in the version order.
    \item $T_i \xrightarrow{wr} T_j$ if $T_i$ writes the version of object $x$ that $T_j$ reads.
    \item $T_i \xrightarrow{rw} T_j$ if $T_i$ reads the version of object $x$ that precedes $T_j$'s write.
    \end{itemize}
  We assume, as does Adya, that if an edge exists between $T_i$ and $T_j$, then
  $T_i \neq T_j$. \fs{What exactly are "i" and "j"? The serialization order, or the timestamps/versions. In paragraph "Acyclicity" we say timestamps, so we should mention it here. Something like: "i and j are the unique timestamps of each transaction"}
  
  First, we prove that if there exists an edge $T_i\xrightarrow{rw/wr/ww} T_j$,
  then $i<j$. We consider each case individually.

  \paragraph{$\xrightarrow{ww}$ case.} Assume that there is a $T_i\xrightarrow{ww} T_j$
  edge. This means that $T_i$ writes a version of object $x$ that precedes $T_j$ in
  the version order. MVTSO's version order for each object is equivalent to the
  timestamp order of the transactions that write to the object. Note that the order
  in which versions are added to the list may not match the timestamp order, but
  the versions are added to the appropriate indices in the list such that the
  timestamps of all preceding versions are smaller and the timestamps of all
  subsequent versions are larger than the new version's timestamp. This implies
  that $i < j$. \fs{I find this sentence hard to read.}

  \paragraph{$\xrightarrow{wr}$ case.} Assume that there is a $T_i\xrightarrow{wr} T_j$
  edge. This means that $T_j$ reads a version of object $x$ written by $T_i$.
  MVTSO's algorithm returns, for a read in $T_j$, the latest version of an object whose
  version timestamp is lower than $j$. This implies that $i < j$.

  \paragraph{$\xrightarrow{rw}$ case.} This case is the most complex. We
  prove it by contradiction. Assume that there is a $T_i \xrightarrow{rw} T_j$ 
  edge such that this edge is the first edge where $i \geq j$.
  $T_i \xrightarrow{rw} T_j$ means that $T_i$
  reads the version of object $x$ that precedes $T_j$'s write. Let that version be $x_k$.
  By the assumption that $T_i \neq T_j$, $i > j$. By the definition of the
  direct serialization graph, $T_k \xrightarrow{wr} T_i$ and $T_k \xrightarrow{ww} T_j$.
  The previous cases imply that $k < i$ and $k < j$.

  Consider a replica that ran the MVTSO-Check for both $T_i$ and $T_j$. There are
  two subcases: either the check for $T_i$ was executed before the check for $T_j$
  or vice versa.
  \begin{enumerate}
    \item \textbf{$T_i$ before $T_j$.~~} 
      $T_i$ must have committed because it exists in the DSG. This implies that
      $T_i$ was in the $\mathit{Prepared}$ set when the check was executed for
      $T_j$ (Line 14 of Algorithm~\ref{a:MVTSO}). When the check for $T_j$ reached
      Line 10 of Algorithm~\ref{a:MVTSO} for $T_j$'s write of $x_j$, the condition was satisfied by
      $T_i$'s read of $x_k$ because $k < j$ and $j < i$. Therefore, the check
      for $T_j$ returned Vote-Abort. However, this is a contradiction because
      $T_j$ committed. \fs{maybe clarify that we are only considering a single replica: I.e. $T_j$ "committed" because this replica voted Vote-Commit. Maybe change committed to Vote-Commit}
    \item \textbf{$T_j$ before $T_i$.~~}
      $T_j$ must have committed because it exists in the DSG. This implies that
      $T_j$ was in the $\mathit{Prepared}$ set when the check was executed for
      $T_i$. When the check for $T_i$ reached Line 7 of Algorithm~\ref{a:MVTSO} for
      $T_i$'s read of $x_k$, the condition was satisfied by $T_j$'s write of $x_j$
      because $k < j$ and $j < i$. Therefore, the check for $T_i$ returned
      Vote-Abort. However, this is a contradiction because $T_i$ is committed.
      \fs{same here. Emphasize at the start of the whole section that we are considering a single replica = unreplicated scenario}
  \end{enumerate}
  
  In all cases, the preliminary assumption leads to a contradiction.
  This implies that $i < j$.

  Next, we use the fact that if there exists an edge
  $T_i\xrightarrow{rw/wr/ww} T_j$, then $i<j$ to prove that the set of
  transactions for which MVTSO-Check returns Commit is serializable.

  \paragraph{Acyclicity.} The set of transactions is serializable if the DSG has no cycles. Assume for
  a contradiction that there exists a cycle consisting of $n$ transactions $T_{\mathit{ts}_1}$,
  ..., $T_{\mathit{ts}_n}$. By the previous fact, this implies that $\mathit{ts}_1 < ... < \mathit{ts}_n < \mathit{ts}_1$.
  However, transaction timestamps are totally ordered. This is a contradiction.
  Thus, the DSG has no cycles.

\end{proof}

\begin{lemma}
 There cannot exist both an \ccert and a \acert for a given transaction.
\label{proof:logged}
\end{lemma}
\begin{proof}
There are two ways that the client can form a \ccert/\acert: through the single phase fast-path, and through the two-phase slow path. We consider both in turn. 
Recall, that on the fast path a \ccert contains a list of \poner as evidence for each shard, while an \acert contains a list of \poner for a single shard (that voted to abort). On the slow path instead, both a \ccert and an \acert contain only a list of \ptwor only for the logging shard $S_{log}$. The intuition for this difference is that on the fast path the two-phase commit decision is yet to be validated, while on the slow path, the two-phase commit decision was already taken and logged in an effort to reduce redundant logging.

In the following, we show that for all possible combinations of certificates, no \ccert and \acert can co-exist.

\par \textbf{Commit Fast Path, Abort Fast Path.} Assume that a client generates both a \ccert and an \acert for $T$ and that both went fast path. A fast \ccert requires $5f+1$ replicas (commit fast path quorum) to vote commit $T$ on \textit{every} shard, while a fast \acert requires either $1$ vote if a \ccert is present to prove the conflict (abort fast path quorum, case 1), or $3f+1$ votes on a \textit{single} shard otherwise (abort fast path quorum, case 2). 
We distinguish abort fast path quorum cases 1 and 2:
\begin{enumerate}
\item If a \ccert exist for a conflicting transaction, then by definition of a \ccert, at least $3f+1$ (commit slow path quorum) replicas of the shard in question must have voted to commit a transaction $T'$ that conflicts with $T$. Since correct replicas never change their vote, by Lemma \ref{proof:mvtso}, at least $2f+1$ correct replicas must vote to abort $T$. Yet, this scenario would require that on this shard, at least one correct replica equivocated ($2f+1$ correct and $3f+1$ must overlap in at least one correct replica). We have a contradiction.
\item Correct replicas never change their vote, yet this scenario would require that on one shard, at least one correct replica equivocated ($5f+1$ and $3f+1$ must overlap in at least one correct replica). We have a contradiction.
\end{enumerate}

\par \textbf{Abort Fast Path, Commit Slow Path.} 
Assume that a client generates both a \ccert that went slow path and an \acert that went fast path for $T$. A slow \ccert requires $n-f = 4f+1$ matching \ptwor replies from the logging shard. In order for a correct replica to send a commit \ptwor message, it must have received a vote tally of at least $3f+1$ commit votes created on \textit{every} shard.
Instead, a fast \acert requires either $1$ vote if a \ccert is present to prove the conflict (abort fast path quorum, case 1), or $3f+1$ votes on a \textit{single} shard otherwise (abort fast path quorum, case 2). 
We distinguish abort fast path quorum cases 1 and 2:
\begin{enumerate}
\item If a \ccert exist for a conflicting transaction, then by definition of a \ccert, at least $3f+1$ (commit slow path quorum) replicas of the shard in question must have voted to commit a transaction $T'$ that conflicts with $T$. Since correct replicas never change their vote, by Lemma \ref{proof:mvtso}, at least $2f+1$ correct replicas must vote to abort $T$. Yet, this scenario would require that on this shard, at least one correct replica equivocated ($2f+1$ correct and $3f+1$ must overlap in at least one correct replica). We have a contradiction.
\item Correct replicas never change their vote, yet this scenario would require that on one shard, at least one correct replica equivocated ($3f+1$ and $3f+1$ must overlap in at least one correct replica). We have a contradiction.
\end{enumerate}

\par \textbf{Commit Fast-Path, Abort Slow-Path.}  Assume that a client generates a \ccert that went fast path, and a \acert that went slow path for a transaction $T$. A fast \ccert requires $5f+1$ replicas (commit fast path quorum) to vote to commit $T$ on \textit{every shard}, while a slow \acert requires $n-f = 4f+1$ \ptwor messages from the logging shard with the decision to abort. In order for a correct replica to send an abort \ptwor message, it must have received a vote tally of at least $f+1$ abort votes created on a single shard.
Correct replicas never change their vote, yet this scenario would require that on one shard, at least one correct replica equivocated ($5f+1$ and $f+1$ must overlap in at least one correct replica). We have a contradiction.

\par \textbf{Commit Slow-Path, Abort Slow-Path.} Assume that a client generates a \ccert and an \acert, both of which went slow path. Recall that each slow path certificate requires $n-f = 4f+1$ matching \ptwor replies from the logging shard. As correct replicas never
change their decision, assembling two sets of $n-f$ matching \ptwor replies would require correct replicas to equivocate. We have a contradiction.

In all cases, we show that there cannot be both a \ccert and \acert for a given transaction $T$. 

\end{proof}

From Lemma \ref{proof:logged} it follows that no two correct replicas can ever process a different outcome (commit/abort) for a transaction $T$. Thus, given a fixed set of total transactions, all replicas are eventually consistent.\\

\begin{lemma}
 If $T_i$ has issued a \ccert and $T_j$ conflicts with $T_i$, then $T_j$ cannot issue a \ccert.
\label{proof:visible}
\end{lemma}

A transaction $T_j$ conflicts with $T_i$ if adding $T_j$ to the serialization graph would create a cycle. If $T_i$ is in $Commit \cup Prepared$ at a replica, the MVTSO concurrency control check
will return abort (Lemma~\ref{proof:mvtso}).
\begin{proof}

    As $T_i$ has committed, a client has generated a \ccert for $T_i$.

    If a client generates a \ccert for $T$, then \textit{at every involved shard} at least $3f+1$ replicas voted to commit the transaction (the slow path requires 3f+1 votes, the fast path 5f+1), and consequently at least $2f+1$ correct replicas voted to commit.

    Assume by way of contradiction that a client $C$ has created a \ccert for $T_j$. There are two ways of achieving this: either $C$ generates
    a \ccert through the fast path or it does so through the slow path. Note, that since $T_i$'s and $T_j$' conflict, their involved shards must intersect in one common shard $S_{common}$.

    \par \textbf{Fast Path.} If this \ccert was created through the fast path, then there must exist at least $4f+1$ correct replicas on $S_{common}$ that voted to commit $T_j$.  But, at least $2f+1$ correct replicas have already voted to commit $T_i$, and would thus vote to abort $T_j$. We have a contradiction. $C$ cannot generate a \ccert for $T_j$ through the fast path.

    \par \textbf{Slow Path.} If this \ccert was created through the slow path, then there must exist at least $2f+1$ correct replicas  on $S_{common}$ (out of $3f+1$) that voted to commit $T_j$ (\poner). But, at least $2f+1$ correct replicas had already voted to commit $T_i$ (and thus should vote to abort $T_j$). As there are only a total of $4f+1$ correct replica, this means that one correct replica voted to commit $T_j$ despite it creating a cycle in the serialization graph with $T_i$ (violating Lemma~\ref{proof:mvtso}). We have a contradiction. $C$ cannot generate a \ccert for $T_j$ through the slow path.

\end{proof}

\begin{theorem} 
\sys{} maintains \textit{Byzantine-serializability}.
\label{proof:ser}
\end{theorem}

\begin{proof}
    Consider the set of transactions for which a \ccert could have been assigned. Consider a
    transaction $T$ in this set. By Lemma~\ref{proof:logged}, there cannot exist an \acert for this transaction. By Lemma~\ref{proof:visible}, there cannot exist a conflicting transaction $T'$ that generated a \ccert. Consequently, there cannot exist a committed transaction $T'$ in the history.  The history thus generates an acyclic serialization graph. The system is thus Byz-Serializable.
\end{proof}

\subsection{Byzantine Independence}

\begin{theorem}
\sys{} maintains Byzantine independence in the absence of network adversary.
\label{proof:independence}
\end{theorem}

We show, that once a client submits a transaction for validation, the transaction's result cannot be unilaterally decided by any group of (colluding) Byzantine participant, be it client or replica.
\begin{proof}
First, we observe that a Byzantine client may never independently choose the result of a transaction. It requires evidence (through \vtaly and \dcert's) supplied by replicas. 

Second, and as a direct consequence, once a client
has prepared a transaction, it cannot unilaterally cause dependent transactions to abort.

Third, both slow abort and slow commit quorums contain at least one correct replica. Fast-path commits/aborts also necessarily contain at least one correct replica. Consequently, if a transaction
commits (or aborts), at least one correct node voted to commit (or abort). A set of Byzantine replicas cannot, on their own, decide the outcome of a transaction.

Finally, a set of Byzantine replicas cannot force a client to read an 1) imaginary read
2) a stale read. Specifically, a client reads committed writes only if there are associated
with a valid \ccert. Moreover, clients read from the $f+1$th latest read version returned. This
ensures that the returned version is at least as recent as what could have been written by an
correct node. Likewise, a client only reads uncommitted writes if they were returned by $f+1$ replicas.
\end{proof}

We note that this is not sufficient to guarantee Byzantine independence when an adversary controls the network. A network adversary, for instance, could systematically inject and reorder transactions at all replicas such that
desired transactions abort.

\subsection{Fallback Safety}

In the absence of an elected leader (\textit{Common Case} of the Fallback protocol), clients use \pone or \ptwo messages and follow the same rules as during normal execution. Byzantine clients that equivocate during \ptwo are indistinguishable from concurrently active clients. As such, all theorems introduced so far in \ref{sec:proofs} continue to hold.

We further note that a fallback leader is only elected when a client detects inconsistent \ptwor messages (the \textit{Divergent Case} of the Fallback protocol), i.e., when a transaction decision was logged on the slow path. Thus we only need to consider the effects of the fallback protocol on decisions that were generated through the slow path. 
Moreover, slow-path \ccert and \acert consist exclusively of messages sent by replicas on the logging shard, and fallback replica leaders only operate on the logging shard). We thus consider only the logging shard in the rest of our proofs.  We show that \one if a slow path \ccert or \acert already exists, the fallback protocol never produces
a conflicting decision certificate, and that \two if no decision certificate exists, any decision \ccert/\acert created
must correspond to a decision made by some client (i.e. the decision is based on a set of \vtaly).

For convenience, we re-state the decision reconciliation rule used by the fallback: the fallback replica proposes a decision \textit{dec$_{new}$ $=$ maj(\{\fbl.decision\})}. We note that matching views are not required here.
 
We start by showing: 

\begin{lemma}
Fallback leaders cannot propose decisions that contradict existing slow path \ccert/\acert. 
\label{proof:durable}
\end{lemma}

\begin{proof}
First, we note that, if a correct replica sends an \fbl message for view $v$, it has necessarily adopted a view $v'\geq v$ and thus will never accept
a decision for a smaller view (1). 

For any existing slow path \ccert/\acert, the associated decision must have been adopted and logged by at least $3f+1$ correct replicas (through
\ptwo and \ptwor messages) in a matching decision view $v$. Without loss of generality, let the decision be Commit (the reasoning for Abort is
identical) and the corresponding decision certificate be a \ccert. Let $R_{FB}$ be the \textbf{first} fallback leader to be elected for a view $v'>v$. 
At least $3f+1$ correct replicas must be in $v'$, since $4f+1$ \fbl messages are required for a leader to be elected in view $v'$ (2). 

We note first that any existing \ccert with view $v \le v'$ must have been constructed from \ptwor messages sent by correct replicas \textbf{before} moving to view $v'$. 
This follows directly from (1): correct replicas do not accept a decision for a smaller view. As $4f+1$ matching replies are necessary to form the \ccert, it
follows that at least $3f+1$ correct nodes have responded Commit in the \ptwor message (3). 

Correct replicas never change their vote. By (1) we have that $3f+1$ correct replicas must be in view $v'$ and by (2) and (3) that $3f+1$ correct replicas voted Commit in
view $v < v'$. It follows that at least $2f+1$ correct replicas provided a Commit decision to the fallback leader $R_{FB}$ in view $v'$. 

Given the decision reconciliation rule, and the fact that the fallback leader is guaranteed to receive at least $2f+1$ Commit decisions, the only new decision that $R_{FB}$ in view $v'$ may propose is therefore Commit.

By induction, this holds for all consecutive views and fallback leaders: if there ever existed $3f+1$ correct replicas that logged decision $dec_v = d$ in view $v$, then those same replicas will only ever log $dec_{v'} = d$ for views $v' > v$, since a fallback leader will always see a majority of decision $d$.

Consequently \sys{} fulfills Lemma \ref{proof:durable}.
\end{proof}

Next, we show:

\begin{lemma}
Any decision proposed by a fallback leader was proposed by a client.
\label{proof:validity}
\end{lemma}

\begin{proof}

Let $R_{FB}$ be the \textbf{first} elected fallback leader that proposes a decision (let its associated view be $v$).

By design, a correct replica only sends a message \fbl once it has logged a decision. Thus $R_{FB}$ is guaranteed to receive a quorum of $4f+1$ \fbl messages
all containing decisions. As the fallback waits for $4f+1$ messages, one decision must be in the majority. By the decision reconciliation rule, $R_{FB}$ proposes the majority decision. 

By assumption, $R_{FB}$ is the first fallback leader to propose a decision. Thus, all decisions included in correct replicas' \fbl messages were \one made by a client, and \two are consistent
with the provided \vtaly (correct replicas will verify that \ptwo messages have sufficient evidence for the decision).

Since any majority decision comprises at least $\geq 2f+1$ \fbl messages, it follows that at least one was created by a correct client, and is hence valid. Consequently, any decision that $R_{FB}$ can propose must have been issued by a client (and is valid).

It follows that all correct replicas will receive valid decisions. In view $v+1$, all decisions from correct replicas forwarded to the next fallback leader will also
be valid. The same reasoning held above regarding the decision rule applies. By induction it follows that for all future views and respective fallback leaders, any (valid) proposed decision must have been proposed by a client.

\underline{Aside}: Note, that if $R_{FB}$ instead is Byzantine, it may collect two \fbl message quorums with different majorities, and equivocate by sending different (valid) \fbd messages to different replicas. In this case, different correct replicas may not adopt the same decision (thus precluding the generation of a decision-certificate), but any decision is nonetheless valid as it was originally proposed by some client.

\end{proof}

We conclude our proof:

\begin{theorem}
Invoking the fallback mechanism preserves Theorem~\ref{proof:ser} and Theorem~\ref{proof:independence}.
\label{proof:safety}
\end{theorem}

\begin{proof}
Lemma \ref{proof:validity} states, that a fallback leader can only propose decisions that were proposed by clients.

Lemma \ref{proof:durable} additionally guarantees, that once slow path \ccert/\acert exist, the fallback mechanism cannot change them. It follows that Lemma \ref{proof:logged} still holds.
Consequently, since any \ccert/\acert generated through the fallback mechanism are indistinguishable from normal case operation, Theorems \ref{proof:ser} and \ref{proof:independence} remain valid.
\end{proof}

\subsection{Fallback Liveness}
\label{proofs:fb_liveness}

We first show that during sufficiently long synchronous periods, the election of a correct fallback leader succeeds. Concretely, we say that after some unknown \textit{global synchronization time} (GST), an upper bound $\Delta$ hold for all message delays.

We note that replicas enforce exponential time-outs on each new view: a replica will not adopt a new view and start a new election until the previous view leader (whether client(s) or fallback replica) has elapsed its time-out.

For convenience, we re-iterate the  view change rules (\S \ref{sec:recovery}, box 2):
\begin{itemize}
\item \textbf{\textit{R1}}: If a view $v$ appears at least $3f+1$ times among the \textit{current views} received in {\em InvokeFB} (and $v$ is larger than the replicas current view), then the replica adopts a new current view $v_{new} = v + 1$. 
\item \textbf{\textit{R2}}: Otherwise, it sets its current view to the largest view that appears at least $f+1$ times among \textit{current views} and is larger than the replicas current view. 
\end{itemize}
When counting how frequently a view is present in \textit{current views}, R uses vote subsumption: the presence of view $v$ counts as a vote also for all $v' \leq v$.

\begin{lemma}
At any time there are at least $2f+1$ correct replicas that are at most one view apart.
\label{proof:invariant}
\end{lemma}
\begin{proof}
A client must provide $3f+1$ matching view responses in the {\em InvokeFB} message (rule \textit{R1}) in order for replicas to adopt the next view and send a \fbl message. This implies that if a correct replica is currently in view $v$, there must exist at least $2f+1$ correct replicas in a view no smaller than $v-1$.
\end{proof}

\begin{theorem}
A correct client can reconcile correct replicas' views in at most two round-trips and one time-out.
\label{proof:bounds}
\end{theorem}

\begin{proof}
Let $v$ be the highest current view of any correct replica. Since there are $n=5f+1$ total replicas, an interested client trying to \textit{InvokeFB} can wait for at least $n-f=4f+1$ replica replies. If the client receives $3f+1$ matching views for a view $v' \geq v-1$ it can use \textit{R1} to propose a new view $v'' = v'+1 \geq v$ that will be accepted by all correct replicas in just a single round-trip.

If a client cannot receive such a quorum, e.g. due to temporary view inconsistency, it must reconcile the views first. This is possible in a single additional step: By Lemma \ref{proof:invariant} any set of $4f+1$ replica responses must contain at least $f+1$ correct replicas' votes for a view $v' \geq v -1 $. Using \textit{R2} a replica may skip ahead to $v'$. Thus, in a second round, the client will be able to receive $\geq 4f+1$ replica replies for a view $\geq v')$, enough to apply \textit{R1} and move all replicas to a common view $\geq v$. We point out, that while only two round-trips of message delays are required, a client may still have to wait out the view time-out between $v-1$ and $v$ (where $v$ is the highest view held by any correct replica).
 
Thus it follows, that a correct client requires at most two round-trips and one time-out to bring all correct replicas to the same view. 
\end{proof}

\underline{Note:} If only a few replicas are in the highest view (i.e. less than $f+1$ -- otherwise which we could potentially catch up in a single round-trip using \textit{R2}), and $4f+1$ replicas are in the same view after using \textit{R2}, then a client does not even need another roundtrip (and a time-out) in order to guarantee the successful leader election (if all $4f+1$ replicas send an \fbl message. If some are Byzantine, this might not happen of course).

Next, we show that during sufficiently long periods of synchrony, a correct fallback is eventually elected.

\begin{lemma}
After GST, and in the presence of correct interested clients, a correct fallback leader is eventually elected.
\label{proof:correct_leader}
\end{lemma} 

\begin{proof}
In the presence of a correct interested client, Byzantine clients cannot stop the successful election of a new Fallback by continuously invoking a view-change on only a subset of replicas. This follows straightforwardly from the fact that after GST, once the time-outs for views grow large enough to fall within $\Delta$, a correct client will \one bring all correct replicas to the same view, and \two there is sufficient time for the fallback replica to propose a decision before replicas move to the next view (and consequently reject any proposal from a lower view).

Moreover, election is non-skipping, as a correct client will broadcast a new-view invocation to all replicas. 

Since fallback leader election is round-robin, it follows that a correct fallback replica will be elected after at most $f+1$ view changes, with a sufficiently long tenure to reconcile a decision across all correct replicas.
\end{proof}~\\

\begin{theorem}
A correct client eventually succeeds in acquiring either a \ccert or \acert for any transaction of interest.
\label{proof:live}
\end{theorem}

\begin{proof}
First, we note, that a timely client can trivially complete all of its own transactions that have no dependencies. However, if a client is slow, or its transaction has dependencies, it may lose autonomy over its own transaction. For a given client c, we define the set \textit{Interested$_c$} to include its own transactions and all their dependencies, as well as any other arbitrary transactions whose completion a client is interested in. 

We distinguish two cases for each $TX \in Interested_c$ that has timed out on its original client: 
\one An interested client manages to receive a \ccert/\acert by either issuing a \rp message and receiving a Fast-Path Threshold of \poner messages for all involved shards, or by issuing a new \ptwo message and receiving $n-f$ \ptwor messages from the logging shard. In this case, a client is able to complete the transaction independently as any client may broadcast \ccert's/\acert's to the involved shards. Theorems \ref{proof:ser} and \ref{proof:independence} are maintained as this case follows the normal-case protocol operation.

\two An interested client cannot obtain decision certificates and starts a Fallback invocation. By Lemma \ref{proof:correct_leader}, during a sufficiently long synchronous period and the presence of an interested client a correct Fallback replica will be elected after at most $f+1$ election rounds. Such a Fallback replica will reconcile a consistent decision across all correct replicas in the log-shard, thus allowing the interested client to receive a $n-f$ matching \ptwor messages. This allows it to construct the respective \ccert/\acert, allowing it to complete the Writeback Phase. 

\end{proof}

\subsection{Revisiting Vote Subsumption}
\label{proofs:subsumption}
Our discussion of the recovery protocol has thus far relied on the idea of \textit{vote subsumption} for validating and choosing a new current view when invoking a fallback leader election. Vote subsumption allows a client to send an {\em InvokeFB} message using a set of \textit{views} that includes replica signatures cast for \textbf{different} current views: when counting how frequently a view is present in the received current views, a replica considers the presence of view $v$ as a  vote also for all $v' \leq v$.

As a consequence, while \textit{aggregate signatures} \cite{
gentry2006identity, boneh2003aggregate} can (always) be applied to compress signatures, more efficient \textit{multi-signatures} \cite{boneh2018compact, micali2001accountable, boldyreva2003threshold, itakura1983public} or \textit{threshold signatures} \cite{shoup2000practical, cachin2005random, boldyreva2003threshold} cannot, since they rely on matching signed messages \changebars{}{(\S~\ref{s:opt})}. In the following we briefly show that, while the current \sys{} prototype uses it to simplify both its protocol and its implementation, vote subsumption is in fact not necessary for progress. Clients, by reasoning carefully about the set of current replica views received, can always wait for matching responses. Thus, \sys{} is indeed capable of leveraging all of the above signature aggregation schemes.

Before continuing the discussion on vote subsumption, we briefly submit the following tangential optimization:
\par \textbf{Aside (Optimization):} A replica in view 0 can accept an \textit{InvokeFB} message without a proof. Thus, a client that proposes view 1 need not attach a set \textit{views} of replica current views and signatures.

Since replicas only move without proof from view 0 to view 1 (and the view number of every replica is at least 0), this optimization trivially maintains the invariant that a majority of correct replicas are no more than one view apart (Lemma \ref{proof:invariant}). Given synchrony, absence of Byzantine fallback leaders, and presence of a correct interested client, only a single fallback invocation is necessary, requiring no views or signatures at all. 

Next, we show that in the general (non-gracious~\ref{clement09aardvark}) case a client can indeed always receive a matching set of current view replies without live-locking.

\begin{lemma}
An interested client trying to send $InvokeFB$ can always use matching current views.
\label{proof:matching}
\end{lemma}

\begin{proof}
There are $n=5f+1$ total replicas, and hence an interested client (trying to submit \textit{InvokeFB}) can wait for at least $4f+1$ current view messages from different replicas.

Let view $v$ be the largest view among the received views. If $v$ was sent by a correct replica, then it follows from Lemma \ref{proof:invariant} that there \textbf{must} be at least $f$ other views $v' \geq v - 1$ among the received views. If there are not, then the client can conclude that $v$ must have been sent by a Byzantine replica, and hence it can remove $v$ (and the replica that sent it) from consideration and wait for an additional current view message. 
We assume henceforth that $v$ is the largest observed view that meets the above criterion. Thus, the client received at least $f+1$ current views ($v$ and $v'$) that are at most one view apart from one another. We can distinguish three cases:

1) There are $3f+1$ matching views: The client can use rule \textit{R1} to send an $InvokeFB$ message using only matching views. \qed

2) There are $f+1$ matching views: The client can use rule \textit{R2} to send an $InvokeFB$ message using only matching views. \qed

3) There are fewer than $f+1$ matching views. However, as established above, there at least $f+1$ current views ($v$ and $v'$) that are at most one view apart from one another. We can distinguish two sub-cases:

\begin{itemize}
\item \underline{$v$ was sent by a correct replica:} Then it follows from Lemma \ref{proof:invariant} that at least $2f+1$ total correct replicas must be in a view $v'' \geq v-1$. Thus, if it does not observe at least $f+1$ matching current view messages for either view $v$ or $v'$, then  the client can keep waiting. By the pigeonhole principle, it must receive $f+1$ matching current views for $v$ or $v'$.

\item \underline{$v$ was not sent by a correct replica:} Then it follows that $v'$ must have been sent by a correct replica, since $f+1$ replicas reported either current view $v$ or $v'$. It follows from Lemma \ref{proof:invariant} that there are at least $2f+1$ correct replicas in view $v'' \geq v' - 1$, and the client can keep waiting until it observe either $f+1$ matching current views for $v’$ or for $v''$.
\end{itemize}

In both cases the client is able to wait for $f+1$ matching current views, and hence it can use \textit{R2} to send an $InvokeFB$ message using only matching views. 
\end{proof}

Since the client can always use matching messages, it follows that both multi-signature and threshold signature schemes are applicable to \textit{InvokeFB} messages as well.

Lemma \ref{proof:matching} does not suffice to conclude progress, since we only showed that a client can always succeed in receiving enough matching replies to apply \textit{R2}.  However, it follows straightforwardly that a client will also be able to apply rule \one within at most another exchange, and thus ensure progress for continued fallback leader election.

\begin{theorem}
Vote subsumption is not necessary for progress in \sys. 
\end{theorem}

\begin{proof}
Since clients are guaranteed to gather sufficient catch up messages (Lemma \ref{proof:matching}), it is guaranteed that, after catching up, there will be at least  $4f+1$ correct replicas within at most one view of each other. Let $v_{max}$ be the larger of the two views. By the pigeonhole principle, a client is guaranteed to be able to wait for least $f+1$ matching votes for $v_{max}$ or $3f+1$ matching votes for $v_{max} -1$. Thus, the client can move all replicas to view $v_{max}$ (using only matching views), ensuring progress.
\end{proof}

\underline{Note:} A client may ensure progress (successful leader election) during a “stable” timeout (after GST) where replicas are not changing views while the client is waiting for replies. If this is not the case, then waiting for the "first" received message from a replica is potentially insufficient to receive matching replies: A client can (and needs to) keep waiting until it \textit{does} receive matching messages, potentially displacing old replica votes with newer ones. Eventually, given partial synchrony, there will be a long enough timeout to ensure stability and progress.  Replicas may simply send their new view to each interested client everytime they increment it, or, clients may "ping" replicas periodically after a timeout to query their latest view.

}{}

\end{document}